%% file: UncertaintyInRanking.tex
\documentclass[12pt]{article}
\usepackage{fullpage}
\usepackage{microtype}
\usepackage{multirow}
\usepackage{amsmath, amsfonts, amssymb, amsthm}
\usepackage{mathtools}
\usepackage{bm}
\usepackage{tikz}
\usepackage[section]{algorithm}
\usepackage{algpseudocode}
\usepackage{verbatim}
\usepackage{hyperref}
\usepackage{subcaption}
\usepackage{cleveref}
\usepackage[sort]{natbib}
\usepackage{complexity}
\usepackage{authblk}
\usepackage{enumerate}
\usepackage[title]{appendix}
\usepackage{adjustbox}

\usetikzlibrary{arrows}
\usetikzlibrary{shapes}

\newcommand{\true}{\operatorname{true}}
\newcommand{\false}{\operatorname{false}}
\newcommand{\sharpp}{\#\P}
\newcommand{\return}[1]{\textbf{return }#1}
\newcommand{\RR}{\mathcal{R}}
\newcommand{\downset}[1]{\downarrow #1}
\newcommand{\upset}[1]{\uparrow #1}

\newcommand{\package}[2]{\texttt{#1}~\citep{#2}}
\newcommand{\idx}{\operatorname{index}}
\newcommand{\ind}[1]{\mathbf{1}(#1)}
\newcommand{\dual}[1]{#1_{\operatorname{op}}}
\newcommand{\predset}[1]{\operatorname{pred}(#1)}
\newcommand{\succset}[1]{\operatorname{succ}(#1)}
\newcommand{\LE}[1]{\operatorname{LE}(#1)}
\newcommand{\keywordname}{\textbf{Keywords:}}

\newcommand{\keywords}[1]{\par\addvspace\baselineskip\noindent\keywordname\enspace\ignorespaces#1}

\numberwithin{figure}{section}
\numberwithin{table}{section}

\newtheorem{theorem}{Theorem}
\numberwithin{theorem}{section}

\newtheorem{lemma}[theorem]{Lemma}

\newtheorem{conjecture}[theorem]{Conjecture}

\begin{document}

\title{Uncertainty in Ranking}
\author{Justin Rising\thanks{The views expressed are those of the author and do not necessarily reflect the official policy or position of the Department of Air Force, the Department of Defense or the U.S. government.}}
\affil{justin.rising@us.af.mil\\Kessel Run\\Boston, MA}
\date{}
\maketitle

\begin{abstract}
Ranks estimated from data are uncertain and this poses a challenge in many applications.  However, estimated ranks are deterministic functions of estimated parameters, so the uncertainty in the ranks must be determined by the uncertainty in the parameter estimates.  We give a complete characterization of this relationship in terms of the linear extensions of a partial order determined by interval estimates of the parameters of interest.  We then use this relationship to give a set estimator for the overall ranking, use its size to measure the uncertainty in a ranking, and give efficient algorithms for several questions of interest.  We show that our set estimator is a valid confidence set and describe its relationship to a joint confidence set for ranks recently proposed by Klein, Wright \& Wieczorek.  We apply our methods to both simulated and real data and make them available through the \texttt{R} package \texttt{rankUncertainty}.
\end{abstract}

\section{Introduction}
\label{sec_intro}

Estimating the ranks of unknown parameters is a fundamental statistical problem with applications in every field that works with data.  When faced with this problem, practitioners typically report the sample ranks of the point estimates with no associated measure of uncertainty.  However, as shown in~\cite{hall2010variability} and~\cite{zuk2007uncertainty}, sample ranks are highly sensitive to measurement noise, so simply reporting a point estimate can vastly overstate how certain we are regarding its value.

Although the need to consider the joint uncertainty of parameter estimates has been recognized since~\cite{louis1984estimPop}, the vast majority of the literature to date has been focused on confidence intervals for individual ranks.  Frequentist methods in the literature fall into two categories: confidence intervals based on the bootstrap and other computational methods~\citep{goldstein1996tables, hall2009bootstrap, marshall1998tables, wright2014ranking, xie2009ties, zhang2014rankCIs}, and confidence intervals based on differences in means~\citep{bie2013ranks, holm2013ranks, mohamad2018cis, lemmers2007ranks}.

The first method to produce a joint measure of uncertainty for sample ranks was introduced in~\cite{klein2020jointCR}.  The key insight there is that the uncertainty in the ranks is determined by the uncertainty in the parameters.  The authors use that insight to build a joint confidence region for the true ranking in terms of simultaneous confidence intervals for the individual ranks.  In this paper, we build on that insight and give a full characterization of the relationship between parameter uncertainty and rank uncertainty, and use that characterization to develop a theory of uncertainty in ranking.

We will illustrate our main idea with a simple example.  Figure~\ref{fig_simple} displays interval estimates for three distinct parameters $\theta_1, \theta_2, \theta_3$.  Assuming that each interval contains the true value of its parameter, we can conclude that $\theta_1 < \theta_3$, but we cannot conclude anything about the rank of $\theta_2$.  We may have $\theta_2 < \theta_1 < \theta_3$, $\theta_1 < \theta_2 < \theta_3$ or $\theta_1 < \theta_3 < \theta_2$.

\begin{figure}[t]
\centering
\input{tex/pex.tex}
\caption{Interval estimates for three parameters.}
\label{fig_simple}
\end{figure}
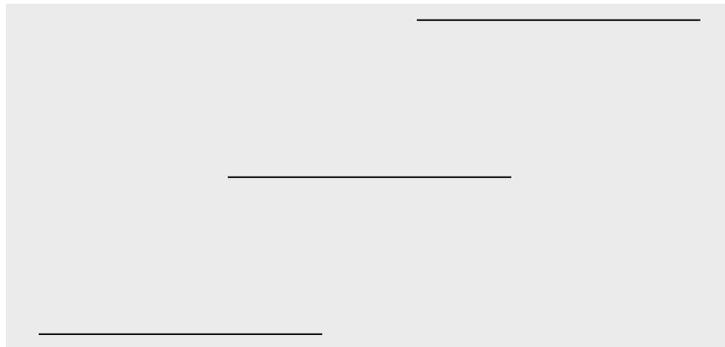

We have a set of rankings that we cannot rule out given the data, and it is natural to take the proportion of possible rankings that belong to this set as a measure of the uncertainty in our ranking.  Measuring this proportion involves counting linear extensions of a particular partial order, a problem with applications in fields such as automated planning~\citep{muise2016planRelax}, Bayesian statistics~\citep{niinimaki2016structDiscov}, causal inference~\citep{wallace1996causal}, event mining~\citep{mannila2000sequential}, minimum-comparison sorting~\citep{peczarski2004sorting}, non-parametric statistics~\citep{morton2009convexRank}, and preference reasoning~\citep{lukasiewicz2014probPref}.  The partial order that we consider was used implicitly in~\cite{klein2020jointCR} to define a joint confidence region for sample rankings and explicitly in~\cite{soliman2010ranking} for improved point estimates of sample ranks in a Bayesian context.  Our contribution consists of showing how this order can be used to characterize the relationship between parameter uncertainty and rank uncertainty and developing methods based on this relationship.

The remainder of this paper is laid out as follows.  In section~\ref{sec_bkgnd}, we give the necessary combinatorial background for further discussion.  In section~\ref{sec_theory}, we characterize the relationship between parameter uncertainty and rank uncertainty, give a set estimate for the true ranking, and discuss issues related to measuring the uncertainty in a ranking.  In section~\ref{sec_other}, we discuss some computations related to rankings.  In section~\ref{sec_freq}, we discuss the frequentist properties of our set estimate and compare our set estimate against that found in~\cite{klein2020jointCR}.  In section~\ref{sec_viz}, we discuss how to visualize the uncertainty in a set of rankings.  In section~\ref{sec_sim}, we perform experiments with simulated data to study the behavior of our methods.  Finally in section~\ref{sec_data}, we apply our methods to data from the American Community Survey.

\section{Notation and Definitions}
\label{sec_bkgnd}

In this section we define the notation we will use and give a brief introduction to general relations and partial orders.  We refer the reader to~\cite{gallier2011discrete} and~\cite{davey2002order} respectively for more information on these topics.

We assume that the reader is comfortable with elementary graph theory and algorithm analysis.  We recommend~\cite{bondy2008graphTheory} and~\cite{cormen2009algorithms} as respective introductions to these subject.

\subsection{General Notation}

$[n]$ denotes the set $\{1, \dots, n\}$.  $[]$ denotes an empty list.  Given an interval $I$, we use $\ell(I)$ and $r(I)$ to denote its left and right endpoints respectively.

Given a data set $x_1, \dots, x_n$, the rank of $x_j$ is denoted by $r_j(x_1, \dots, x_n)$ and is equal to $\sum_{i = 1}^n \ind{x_i \leq x_j}$.  The vector whose $j$th entry is $r_j(x_1, \dots, x_n)$ is referred to as a ranking and is denoted by $\mathbf{r}(x_1, \dots, x_n)$.  When there is no potential for confusion, we will elide the data and simply denote a ranking as $\mathbf{r}$.

\subsection{Relations}

Given a set $S$, a binary relation $\RR$ on $S$ is any subset of $S \times S$.  We write $s\,\RR\,t$ as a shorthand for $(s, t) \in \RR$.  If $T \subseteq S$, the restriction of $\RR$ to $T$ is denoted by $\RR {\restriction_T}$ and defined as $\RR \cap (T \times T)$.

If $s\,\RR\,s$ for all $s \in S$, we say that $\RR$ is reflexive.  If $s\,\RR\,t$ and $t\,\RR\,s$ implies that $s = t$, we say that $\RR$ is antisymmetric.  If $s\,\RR\,t$ and $t\,\RR\,u$ implies that $s\,\RR\,u$, we say that $\RR$ is transitive.  Finally, if $s\,\RR\,t$ or $t\,\RR\,s$ for all $s, t \in S$, we say that $\RR$ is total.

The transitive closure of a relation $\RR$ is the smallest transitive relation $\RR_T$ such that $\RR \subseteq \RR_T$.  The transitive reduction of $\RR$ is the smallest relation $\RR_t$ such that the transitive closure of $\RR_t$ is equal to the transitive closure of $\RR$.

\subsection{Partially Ordered Sets}

A partially ordered set (or poset) is an ordered pair $(S, \preceq)$ such that $\preceq$ is a reflexive, antisymmetric and transitive relation on $S$.  $S$ is referred to as the ground set, and $\preceq$ is referred to as a partial order on $S$.  A linear order is a total partial order.

We define the dual of $\preceq$ to be the order $\dual{\preceq}$ such that $s \dual{\preceq} t$ if and only if $t \preceq s$.  The poset $(S, \dual{\preceq})$ is referred to as the dual of $(S, \preceq)$.

We write $s\prec t$ when $s \preceq t$ and $s \neq t$.  The relation $\prec$ is referred to as a strict partial order.  Given a strict partial order $\prec$, we can define a partial order $\preceq$ by $s \preceq t$ whenever $s \prec t$ or $s = t$.  If $s \not\preceq t$ and $t \not\preceq s$, we say that $s$ and $t$ are incomparable and write $s \| t$.

Suppose that $\preceq_1$ and $\preceq_2$ are partial orders on some set $S$.  We say that $\preceq_2$ is an extension of $\preceq_1$ if $s \preceq_1 t$ implies $s \preceq_2 t$.  A linear extension of a partial order $\preceq$ is a linear order $\leqslant$ such that $\leqslant$ is an extension of $\preceq$.  The set of linear extensions of $\preceq$ is denoted by $\LE{\preceq}$.

$T \subseteq S$ is a down-set if $t \in T$ and $s \preceq t$ implies that $s \in T$.  $T$ is a principal down-set if there is some $s \in S$ such that $T = \{x \in S \mid x \preceq s\}$.  In this case we write $T = \downset{s}$.  The predecessor set of $s$, written as $\predset{s}$, is equal to $\downset{s} \setminus \{s\}$.  An up-set of a poset is defined as a down-set in the dual poset, and a principal up-set is written as $\upset{s}$.  Likewise, the successor set of $s$, written as $\succset{s}$, is its predecessor set in the dual poset.  An interval $[s, t]$ is defined to be $\upset{s} \cap \downset{t}$.

$s$ is a minimal element of a poset if $t \preceq s$ implies that $s = t$ and a minimum if $s \preceq t$ for all $t$.  In a linearly ordered set every minimal element is a minimum but this is not true for general posets.

The reachability relation of a graph is the transitive closure of its edges.  The reachability relation of any directed acyclic graph $G$ is a partial order which we will denote $\preceq_G$.  Furthermore, every partial order can be realized as the reachability relation of some directed acyclic graph.  If $\preceq$ is a partial order and $G$ is a graph whose reachability relation is $\preceq$, we say that $G$ is a graph of $\preceq$.  The cover graph of a partial order $\preceq$ is the graph whose set of edges is equal to the transitive reduction of the reachability relation of any graph of $\preceq$.

We will also be interested in a matrix representation of a poset.  This a matrix $M \in \{0, 1\}^{|S| \times |S|}$ such that $m_{ij} = \ind{s_i \preceq s_j}$.

Many of the most interesting results from order theory apply to posets of arbitrary size.  In this paper, however, we will only be interested in finite posets.

\section{From Parameter Uncertainty to Rank Uncertainty}
\label{sec_theory}

In this section we discuss the relationship between parameter uncertainty and rank uncertainty.  We first show how to define rank uncertainty in terms of parameter uncertainty, and then we discuss computational issues related to measuring rank uncertainty.

\subsection{Defining Rank Uncertainty}

We consider the problem of estimating the ranks of several unknown real-valued parameters $\theta_1, \dots, \theta_p$.  $\mathbf{r}(\hat{\theta}_1, \dots, \hat{\theta}_p)$ is a point estimate of $\mathbf{r}(\theta_1, \dots, \theta_p)$, but as always we wish to have some measure of its uncertainty.

In any realistic context, we will have some measure of the uncertainty in each point estimate.  These are typically presented as interval estimates $\hat{I}_1, \dots, \hat{I}_p$.  We will show that these interval estimates also contain information about the uncertainty in the ranking of the parameters.

In the presence of ties in the parameter values, the sample ranks are not a consistent estimator of the population ranks~\citep{xie2009ties}, so we only consider the case where $\theta_{j_1} \neq \theta_{j_2}$ when $j_1 \neq j_2$.  We additionally assume that no two interval estimates share a common endpoint, that that $\hat{\theta}_j \in \hat{I}_j$ for all $j$ and that each $\hat{I}_j$ has positive length.

We say that a ranking $\mathbf{r}$ is compatible with the data if there is a set of points $x_1, \dots, x_p$ such that $x_j \in \hat{I}_j$ for all $j$ and $\mathbf{r} = \mathbf{r}(x_1, \dots, x_p)$.  The sample ranks of the point estimates are compatible with the data, but other rankings are as well unless the interval estimates are pairwise disjoint.  Because we assume that $\theta_{j_1} \neq \theta_{j_2}$ when $j_1 \neq j_2$, we are guaranteed that $\mathbf{r}(\theta_1, \dots, \theta_p)$ has $p$ distinct values, so we only need to consider the $p!$ possible permutations of $[p]$.

We define our set estimator for the true ranking to be the set of rankings which are compatible with the data.   Working directly with the definition of this set is challenging, but its size is analogous to the width of an interval estimate.  The range of possible sizes grows with the number of parameters we wish to estimate, but we can always normalize it by dividing by the number of all possible rankings.  We therefore take this ratio to be our measure of uncertainty.

In order to be able to work with this measure of uncertainty, we will give a characterization of compatibility with the data in terms of order relations.  For this, we will need to define two orders.

We first consider the order generated by a ranking.  For any ranking $\mathbf{r}$, we define $\leqslant_\mathbf{r}$ by $j_1 \leqslant_\mathbf{r} j_2$ whenever $r_{j_1} \leq r_{j_2}$.

We then consider the order generated by a set of intervals.  Given a set of intervals $\mathcal{I}$, we will follow~\cite{klein2020jointCR} and~\cite{soliman2010ranking} and define $\prec_\mathcal{I}$ by $j_1 \prec_\mathcal{I} j_2$ whenever $r(I_{j_1}) < \ell(I_{j_2})$.

In the remainder of this paper, $\mathcal{I}$ will be assumed to be a set of interval estimates for the parameters of interest.  $\leqslant_\mathbf{r}$ and $\preceq_\mathcal{I}$ are then orders on $[p]$.  The relationship between $\leqslant_\mathbf{r}$, $\preceq_\mathcal{I}$ and compatibility in the data is given in Theorem~\ref{thm_lin_ext}, which follows from Lemma~\ref{lem_lin_ext_fwd} and Lemma~\ref{lem_lin_ext_bwd}.

\begin{lemma}
\label{lem_lin_ext_fwd}
If a ranking $\mathbf{r}$ is compatible with the data, then $\leqslant_\mathbf{r}$ is a linear extension of $\preceq_\mathcal{I}$.
\end{lemma}
\begin{proof}
Since $\mathbf{r}$ is compatible with the data, there are points $x_1, \dots, x_n$ such that $x_j \in \hat{I}_j$ for all $j$ and $\mathbf{r} = \mathbf{r}(x_1, \dots, x_n)$.  If $j_1 \| j_2$ for all $j_1, j_2$ then $\leqslant_\mathbf{r}$ is a linear extension of $\preceq_\mathcal{I}$ because every linear order on $[p]$ is a linear extension of $\preceq_\mathcal{I}$.  Otherwise there are $j_1, j_2$ such that $j_1 \prec_\mathcal{I} j_2$.  Then $r(\hat{I}_{j_1}) < \ell(\hat{I}_{j_2})$, and this guarantees that $x_{j_1} < x_{j_2}$.  Therefore $j_1 \leqslant_\mathbf{r} j_2$.  Since this holds for every $j_1, j_2$ with $j_1 \prec_\mathcal{I} j_2$, $\leqslant_\mathbf{r}$ is a linear extension of $\preceq_\mathcal{I}$.
\end{proof}

\begin{lemma}[\cite{trotter1997intervalOrder}]
\label{lem_lin_ext_bwd}
If $\leqslant_\mathbf{r}$ is a linear extension of $\preceq_\mathcal{I}$, then $\mathbf{r}$ is compatible with the data.
\end{lemma}

\begin{theorem}
\label{thm_lin_ext}
A ranking $\mathbf{r}$ is compatible with the data if and only if $\leqslant_\mathbf{r}$ is a linear extension of $\preceq_\mathcal{I}$.
\end{theorem}

In light of Theorem~\ref{thm_lin_ext}, we will denote the set of rankings compatible with the data as $\LE{\preceq_\mathcal{I}}$.  Our measure of uncertainty will then be $|\LE{\preceq_\mathcal{I}}|/p!$.

All of the theory and methods we will develop rely on the assumption that no two intervals share a common endpoint.  Although this will often hold in practice, it will not always be satisfied, and it is reasonable to worry that we will need methods that do not rely on this assumption.  However, this turns out not to be an issue.

The key insight here is that the statistic we are interested in is $\preceq_\mathcal{I}$ rather than $\mathcal{I}$.  If $\mathcal{I}^{'}$ is a set of intervals such that $\preceq_\mathcal{I} = \preceq_{\mathcal{I}^{'}}$, then any computation on $\mathcal{I}$ that depends only on $\preceq_\mathcal{I}$ will yield the same result when we perform it on $\mathcal{I}^{'}$.  Therefore, if we can find a set of intervals $\mathcal{I}^{'}$ such that $\preceq_\mathcal{I} = \preceq_{\mathcal{I}^{'}}$ and $\mathcal{I}^{'}$ has distinct endpoints, then we can use $\mathcal{I}^{'}$ in place of $\mathcal{I}$ to avoid issues related to common endpoints.

A procedure to find such an $\mathcal{I}^{'}$ is given in~\cite{madej1991interval}.  While we can apply this procedure to any set of intervals, for numerical reasons we prefer to start with a set of intervals with integer endpoints.  As shown in~\cite{mitas1994semiRep}, $\preceq_\mathcal{I}$ has a canonical representation with \[I_j = [|\{\predset{j^{'}} \mid \predset{j^{'}} \subset \predset{j}\}|, |\{\succset{j^{'}} \mid \succset{j} \subset \succset{j^{'}}\}|]\] for all $j \in [p]$.  Given a set of intervals $\mathcal{I}$, we can compute this canonical representation and then distinguish the endpoints to produce a set of intervals with distinct endpoints that generates $\preceq_\mathcal{I}$.  Because we need to be able to run this procedure on sets of intervals where the endpoints are not distinct, we requires a matrix representation of $\preceq_\mathcal{I}$.  The time and space complexity are therefore both $O(p^2)$.

\subsection{Measuring Rank Uncertainty}

In light of Theorem~\ref{thm_lin_ext}, we can measure the uncertainty in the ranks of a collection of estimated parameters by counting the linear extensions of the order generated by their interval estimates.  However, this is more easily said than done.  As shown in~\cite{brightwell1991counting}, counting linear extensions of a partial order is intractable\footnote{More precisely, the problem of counting the linear extensions of a given partial order is $\sharpp$-complete.  A discussion of what this means is beyond the scope of this paper.  See~\cite{arora2009complexity} for the relevant definitions.} in general.  Fortunately, there are a number of algorithms available for efficiently approximating the number of linear extensions of a given partial order.  We refer the reader to~\cite{talvitie2018countInPractice} and references therein for details.

Here we introduce the idea of partitioning $\preceq_I$.  This is a method that allows us to count its linear extensions by counting the linear extensions of several smaller orders and multiplying the results.  To state this idea, we will need the notion of an order partition.

An \emph{order partition} of a poset $(S, \preceq)$ is a partition $S_1, \dots, S_k$ of $S$ with the property that $s_{j_1} \prec s_{j_2}$ if and only if $j_1 < j_2$ for all $s_{j_1} \in S_{j_1}$, $s_{j_2} \in S_{j_2}$.  We say that an order partition is maximal if no $S_j$ can be further partitioned.

If we have an order partition of $(S, \preceq)$, then the number of linear extensions of $\preceq$ is equal to the product of the number of linear extensions of $\preceq$ restricted to each component of the partition.  While this does not reduce the asymptotic complexity of counting linear extensions, it can bring the running time down significantly for some posets.

We are not aware of an algorithm for computing a maximal order partition for a general poset.  However, we can compute an order partition of $([p], \preceq_I)$ with Algorithm~\ref{algo_part}.  We prove the correctness of this algorithm and analyze its complexity in Theorem~\ref{thm_algo_part}.

\begin{theorem}
\label{thm_algo_part}
Algorithm~\ref{algo_part} computes a maximal order partition of $([p], \preceq_I)$.  Its time complexity is $O(p\log(p))$ and its space complexity is $O(p)$.
\end{theorem}
\begin{proof}
We must first show that Algorithm~\ref{algo_part} produces an order partition of $([p], \preceq_I)$.  In this case, an order partition is a collection of sets $S_1, \dots, S_k$ such that $r(\hat{I}_{j_1}) < \ell(\hat{I}_{j_2})$ whenever $j_1 \in S_{j_1}, j_2 \in S_{j_2}$ and $j_1 < j_2$.  By construction, the subsets of $[p]$ inserted into $P$ have this property.  To see that the output is a maximal order partition, simply observe that the algorithm will find all cases where the remaining left endpoint are greater than the greatest right endpoint observed so far.

The running time of Algorithm~\ref{algo_part} is dominated by the cost of sorting the endpoints.  We need to store the output and a representation of $\sigma$, so the total space is $O(p)$.
\end{proof}

\begin{algorithm}[t]
\begin{algorithmic}
\Procedure{PartitionOrder}{$\hat{I}_1, \dots, \hat{I}_p$}
\State Sort $\hat{I}_1, \dots, \hat{I}_p$ by their left endpoints
\State Let $\sigma$ be the permutation that maps $\hat{I}_1, \dots, \hat{I}_p$ to their sorted indices
\State $j_* \gets 1$
\State $P \gets []$
\State $r \gets -\infty$
\For{$j \in [p]$}
\State $r \gets \max(r, r(\hat{I}_j))$
\If{$j < p$ and $r < \ell(\hat{I}_{j + 1})$}
\State Append $\{\sigma(j_*), \dots, \sigma(j)\}$ to $P$
\State $j_* \gets j + 1$
\EndIf
\EndFor
\State Append $\{\sigma(j_*), \dots, \sigma(|S|)\}$ to $P$
\State $\return{P}$
\EndProcedure
\end{algorithmic}
\caption{Partition a set of intervals into subsets of disjoint index.}
\label{algo_part}
\end{algorithm}

\section{Other Questions of Interest}
\label{sec_other}

While $|\LE{\preceq_\mathcal{I}}|/p!$ is a measure of uncertainty in a ranking estimate, our inability to efficiently calculate it exactly presents us with some difficulty.  Here we consider two tractable questions related to $\LE{\preceq_\mathcal{I}}$ .  We first discuss how to determine whether a given ranking is an element of $\LE{\preceq_\mathcal{I}}$.  We then discuss which parameters are in the smallest $k$ of some ranking compatible with the data.

\subsection{Is a Given Ranking Compatible with the Data?}

\begin{algorithm}[p]
\begin{algorithmic}
\Procedure{DetermineCompatibility}{$\mathbf{r}, \hat{I}_1, \dots, \hat{I}_p$}
\State $P \gets \Call{Sort}{\{\ell(\hat{I}_1), r(\hat{I}_1), \dots, \ell(\hat{I}_p), r(\hat{I}_p)\}}$
\State $j_* \gets \min(\{j \mid P[j]$ is a right endpoint$\})$
\For{$j \in [p]$}
\If{$P[j_*] < \ell(\hat{I}_j)$}
\State $\return{\false}$
\Else
\State Mark $\ell(\hat{I}_j), r(\hat{I}_j)$ as deleted
\State $j_* \gets \min(\{j \mid P[j]$ is a right endpoint and not marked as deleted$\})$
\EndIf
\EndFor
\State $\return{\true}$
\EndProcedure
\end{algorithmic}
\caption{Decide whether a given ranking is compatible with the data.}
\label{algo_compat}
\end{algorithm}

We first consider the problem of determining whether a given ranking is compatible with the data.  Our method is given in Algorithm~\ref{algo_compat}, and we prove its correctness and analyze its complexity in Theorem~\ref{thm_algo_compat}.

\begin{theorem}
\label{thm_algo_compat}
Algorithm~\ref{algo_compat} correctly determines whether a given ranking is compatible with the data in time $O(p\log(p))$ and space $O(p)$.
\end{theorem}
\begin{proof}
By Theorem~\ref{thm_lin_ext}, $\mathbf{r}$ is compatible with the data if and only if $\leqslant_\mathbf{r}$ is a linear extension of $\preceq_\mathcal{I}$.

If $\leqslant_\mathbf{r}$ is a linear extension of $\preceq_\mathcal{I}$, then there is no pair $j_1, j_2$ such that $r(\hat{I}_{j_1}) < \ell(\hat{I}_{j_2})$ but $r_{j_1} > r_{j_2}$.  In that case the algorithm will never find a right endpoint to the left of the current interval's left endpoint.  As a result, it will return true.

If $\leqslant_\mathbf{r}$ is not a linear extension of $\preceq_\mathcal{I}$, then there is some pair $j_1, j_2$ such that $r(\hat{I}_{j_1}) < \ell(\hat{I}_{j_2})$ and $r_{j_1} > r_{j_2}$.  In that case the algorithm will find a right endpoint to the left of the current interval's left endpoint.  It will then return false.

We need to sort the endpoints to perform the algorithm, so the time complexity is at least $O(p\log(p))$.  We can find each endpoint through binary search, and in total this requires $O(p\log(p))$ steps.  Maintaining $j_*$ requires us to look at each element exactly once, so the total time complexity is $O(p\log(p))$.

We need to store the list $P$, a bit for each interval denoting whether its endpoints have been deleted, a map of the locations of the endpoints of each interval in $P$, and the index $j_*$, so the total space required is $O(p)$.
\end{proof}

\subsection{Which Parameters Can Be in the Smallest \texorpdfstring{$k$}{k}?}

\begin{algorithm}[p]
\begin{algorithmic}
\Procedure{PrincipalDownSets}{$\hat{I}_1, \dots, \hat{I}_p$}
\State $P \gets \Call{Sort}{\{\ell(\hat{I}_1), r(\hat{I}_1), \dots, \ell(\hat{I}_p), r(\hat{I}_p)\}}$
\State $r \gets 1$
\For{$q \in P$}
\State Let $q$ be an endpoint of $\hat{I}_j$
\If{$q$ is a left endpoint}
\State $c_j \gets r$
\Else
\State $r \gets r + 1$
\EndIf
\EndFor
\State $\return{c}$
\EndProcedure
\end{algorithmic}
\caption{Compute the cardinalities of the principal down-sets of $\preceq_\mathcal{I}$.}
\label{algo_down_set}
\end{algorithm}

In many applications of ranking, we wish to select the $k$ smallest unknown parameters.  The standard answer to this question is $\{\hat{\theta}_{(j)} \mid j \leq k\}$, but this ignores the uncertainties in the point estimates.

If instead we assume that $\theta_j \in \hat{I}_j$ for all $j$, we can ask which values can be in the $k$ smallest.  In the example given in Figure~\ref{fig_simple}, the smallest value may be either $\theta_1$ or $\theta_2$, and the two smallest values may be $\{\theta_1, \theta_2\}$ or $\{\theta_1, \theta_3\}$.

However, we are not completely uncertain regarding the smallest values.  The claims that $\theta_3$ is the smallest value and that $\{\theta_2, \theta_3\}$ are the two smallest values are not compatible with the data because $r(\hat{I}_1) < \ell(\hat{I}_3)$.

To solve this problem, we need some new vocabulary.  First, we need the notion of an index.  If $(S, \leqslant)$ is a linearly ordered set and $s \in S$, the index of $s$ with respect to $\leqslant$ is defined to be $|\{t \in S \mid t \leqslant s\}|$ and is denoted $\idx_\leqslant(s)$.

We will also need the notion of a $k$-top set.  Given a partially ordered set $(S, \preceq)$ and $k \leq |S|$, we define the $k$-top set to be $\{j \mid \idx_\leqslant(j) \leq k$ for some linear extension $\leqslant$ of $\preceq\}$.

It is clear that $|\downset{s}| \leq \idx_\leqslant(s)$ for all $s \in S$.  Therefore, the $k$-top set is equal to $\{s \in S \mid |\downset{s}| \leq k\}$, and we can compute the $k$-top set of $(S, \preceq)$ by computing the cardinalities of the principal down-sets of $\preceq$ and selecting those that have at most $k$ elements.  Algorithm~\ref{algo_down_set} computes the cardinalities of these principal down-sets for the order generated by a set of intervals.  We prove the correctness of this algorithm in Theorem~\ref{thm_algo_down_set}.

\begin{theorem}
\label{thm_algo_down_set}
Algorithm~\ref{algo_down_set} computes the cardinalities of the principal down-sets of $\preceq_\mathcal{I}$ in time $O(p\log(p))$ and space $O(p)$.
\end{theorem}
\begin{proof}
To see this, observe that $|\downset{j}|$ is one greater than the number of intervals whose right endpoints are less than $\ell(\hat{I}_j)$.  The counter $r$ is equal to this value when each left endpoint is encountered, so the algorithm is correct.

The running time of Algorithm~\ref{algo_down_set} is dominated by the cost of sorting the endpoints, which will in general be $O(p\log(p))$.  The memory required is determined by the size of $P$ and the cost of storing $c(j)$ for $j \in [p]$, so the total is $O(p)$.
\end{proof}

We can define the $k$-bottom set of $(S, \preceq)$ to be the $k$-top set of $(S, \dual{\preceq})$.  If $\preceq$ is the order generated by a set of intervals, $\dual{\preceq}$ is the order generated by the image of these intervals under the map $x \mapsto -x$.  We can therefore compute the $k$-bottom set of $\preceq_\mathcal{I}$ by applying Algorithm~\ref{algo_down_set} to the negated intervals.

However, it is worth considering how to describe the $k$-bottom set in terms of the original order.  The $k$-bottom set is equal to $\{s \in S \mid |\upset{s}| \leq k\}$.  Inclusion-exclusion allows us to give an alternate formula for $\idx_\leqslant(s)$, which we record as Lemma~\ref{lem_account}.

\begin{lemma}
\label{lem_account}
Let $(S, \leqslant)$ be a linearly ordered set.  Then $\idx_\leqslant(s) = |S| + 1 - |\{t \in S \mid s \leqslant t\}|$.
\end{lemma}

The $k$-bottom set of $(S, \preceq)$ can therefore be described as the set of elements whose index is at most $|S| + 1 - k$ in some linear extension of $\preceq$.

\section{Confidence Sets for Rankings}
\label{sec_freq}

Up until now we have considered general interval estimates, but in this section we turn our attention to the case where $\hat{I}_1, \dots, \hat{I}_p$ are confidence intervals.  We will show that $\LE{\preceq_\mathcal{I}}$ is a confidence set for the true ranking and analyze a similar confidence set found in~\cite{klein2020jointCR}.

\subsection{Coverage Probability of the Set of Rankings Compatible with the Data}

We begin by considering the coverage probability of $\LE{\preceq_\mathcal{I}}$.  Our main result is Theorem~\ref{thm_freq_joint}, which allows us to interpret $\LE{\preceq_\mathcal{I}}$ as a confidence region for the true ranking.

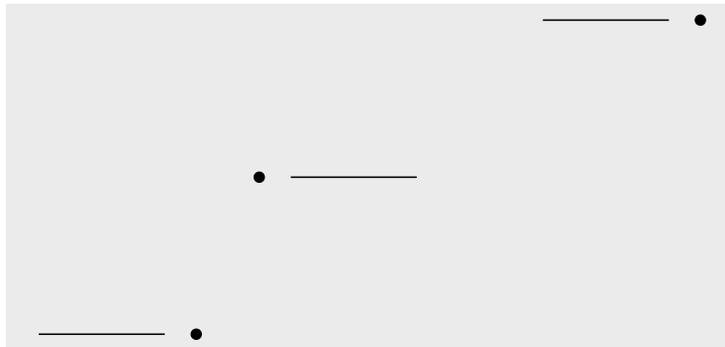
\begin{figure}[t]
\centering
\input{tex/pfreq.tex}
\caption{The true ranking can be compatible with the data even if no confidence interval contains the true value of its parameter.}
\label{fig_freq}
\end{figure}

\begin{theorem}
\label{thm_freq_joint}
Suppose that $\hat{I}_1, \dots, \hat{I}_p$ are confidence intervals with joint coverage probability $1 - \alpha$.  Then $\LE{\preceq_\mathcal{I}}$ contains the true ranking with probability at least $1 - \alpha$.
\end{theorem}
\begin{proof}
If each confidence interval $\hat{I}_k$ contains the true value of the corresponding parameter, then the true ranking is compatible with the data.  This happens with probability at least $1 - \alpha$.
\end{proof}

The true ranking of the parameters may be compatible with the data even when no interval contains the true value of its parameter (see Figure~\ref{fig_freq} for an example).  Therefore the coverage probability of $\LE{\preceq_\mathcal{I}}$ will in general be higher than the joint coverage probability of the original confidence intervals.

\subsection{An Analysis of the Product Confidence Set}

We now turn to the confidence set given in~\cite{klein2020jointCR}.  There the authors gave a probabilistic proof that this set has the correct coverage probability.  We will prove the same result using order theoretic methods, and then show that this region may be more conservative than $\LE{\preceq_\mathcal{I}}$.

The confidence set given in~\cite{klein2020jointCR} is equal to the product of intervals that contain the sample ranks.  In our terms, the interval for $\hat{r}_j$ is equal to $[|\downset{j}|, p + 1 - |\upset{j}|]$.  We know from the discussion in section~\ref{sec_other} that $\idx_\leqslant(j)$ is contained in this interval for any linear extension $\leqslant$ of $\preceq_I$.  Furthermore, for any $i \in [|\downset{j}|, p + 1 - |\upset{j}|]$, we can construct a linear extension $\leqslant$ of $\preceq_I$ such that $\index_\leqslant(j) = i$.  This interval is the set of possible indices of $j$ in all linear extensions of $\preceq_I$, so we will refer to it as the \emph{index interval} of $j$ with respect to $\preceq_I$.

We will show that the index intervals are valid confidence intervals for the individual ranks.  Our analysis hinges on Lemma~\ref{lem_ind_int}.

\begin{lemma}
\label{lem_ind_int}
If a ranking $\mathbf{r}$ is compatible with the data, then $r_j \in [|\downset{j}|, p + 1 - |\upset{j}|]$ for all $j \in [p]$.
\end{lemma}
\begin{proof}
We know from Theorem~\ref{thm_lin_ext} that $\mathbf{r}$ is compatible with the data if and only if $\leqslant_\mathbf{r}$ is a linear extension of $\preceq_I$.  Therefore we assume that $\leqslant$ is a linear extension of $\preceq_I$.  Then we must have $j_1 \leqslant j_2$ whenever $j_1 \preceq_I j_2$, and this implies that $\idx_\leqslant(j) \geq |\downset{j}|$.  By a similar argument and an appeal to Lemma~\ref{lem_account}, we can show that $\idx_\leqslant(j) \leq p + 1 - |\upset{j}|$.
\end{proof}

We can now give a new proof of the main result of~\cite{klein2020jointCR}.

\begin{theorem}
\label{thm_freq}[\cite{klein2020jointCR}]
Suppose that $\hat{I}_1, \dots, \hat{I}_p$ are confidence intervals with joint coverage probability $1 - \alpha$.  Then the index intervals of $\preceq_I$ are simultaneous $1 - \alpha$ confidence intervals for the true ranks.
\end{theorem}
\begin{proof}
This follows from Theorem~\ref{thm_freq_joint}, Lemma~\ref{lem_ind_int} and Theorem~\ref{thm_lin_ext}.
\end{proof}

We note that we can compute the index intervals of $\preceq_I$ by applying Algorithm~\ref{algo_down_set} to $\preceq_I$ and its dual.  The time and space complexity of this operation are $O(p\log(p))$ and $O(p)$ respectively.

\begin{figure}[t]
\centering
\input{tex/pcounter.tex}
\\
\begin{tabular}{|c|c|}
\hline
Parameter&Index Interval\\
\hline
$\theta_1$&[1, 3]\\
$\theta_2$&[1, 4]\\
$\theta_3$&[1, 5]\\
$\theta_4$&[2, 5]\\
$\theta_5$&[3, 5]\\
\hline
\end{tabular}
\caption{A counterexample to the converse of Lemma~\ref{lem_ind_int}.}
\label{fig_counter}
\end{figure}

To show that the product confidence set is conservative, we note that the converse of Lemma~\ref{lem_ind_int} is false.  Consider the intervals shown in Figure~\ref{fig_counter}.  The order $\leqslant$ defined by $3 \leqslant 4 \leqslant 1 \leqslant 2 \leqslant 5$ is a not a linear extension of the order generated by these intervals, but $\idx_\leqslant(j) \in [|\downset{j}|, p + 1 - |\upset{j}|]$ for all $j \in [p]$.  Therefore, while the product confidence set is guaranteed to contain every ranking that is compatible with the data, it can contain rankings that are not compatible with the data.

This conservatism seems to be deeply embedded in the index intervals themselves.  Any intervals which are strictly contained in the index intervals rule out rankings that are compatible with the data.  We find it difficult to imagine that a method that produces narrower intervals would always produce confidence intervals with the correct coverage probability.  However, we do not have a proof of this claim, so we must state it as a conjecture.

\begin{conjecture}
\label{conj_int_width}
No procedure which produces intervals that are strictly contained in the index intervals will always produce valid confidence intervals for the individual ranks.
\end{conjecture}

We also note that in some cases the index intervals may contain precisely the rankings that are compatible with the data.  This is the case for the order generated by the intervals in Figure~\ref{fig_simple}.  We still recommend using $\LE{\preceq_\mathcal{I}}$ in that case because its size directly measures the uncertainty in the ranking.  The size of the product confidence set is correlated with the uncertainty in the ranking, but is harder to interpret, and may even be greater than $p!$ in some cases.

Finally, following~\cite{klein2020jointCR}, we note that the arguments given here regarding coverage probabilities apply to credible regions in the Bayesian setting.  Given a set of credible intervals with joint posterior probability $1 - \alpha$, $\LE{\preceq_\mathcal{I}}$ forms a credible region for the sample ranking with posterior probability at least $1 - \alpha$.  Similarly, the index intervals of $\preceq_\mathcal{I}$ form a credible region for the individual ranks with joint posterior probability at least $1 - \alpha$.

\section{Visualizing the Set of Rankings Compatible with the Data}
\label{sec_viz}

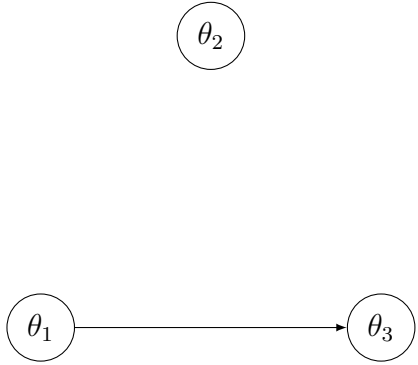
\begin{figure}[t]
\centering
\input{tex/pcover.tex}
\caption{The cover graph of the order generated by the intervals shown in Figure~\ref{fig_simple}.}
\label{fig_viz}
\end{figure}

$|\LE{\preceq_\mathcal{I}}|/p!$ gives us a measure of how much uncertainty is in a ranking, but there is value in being able to visualize that uncertainty.  A plot that depicts every linear extension of $\preceq_\mathcal{I}$ is in general computationally expensive to produce and difficult to interpret, so we need to find some better alternative.

One obvious approach is to plot the intervals themselves, but this is not the only possible choice.  Any graph of $\preceq_\mathcal{I}$ contains the same information as the intervals in $I$.  A graph with fewer edges is more easily understood than one with more, so the cover graph of $\preceq_\mathcal{I}$ is a natural choice.  We can compute the cover graph of $\preceq_\mathcal{I}$ by computing any graph of $\preceq_\mathcal{I}$ and removing any edge $(j_1, j_2)$ such that there is some $j_3$ with $j_1 \prec_\mathcal{I} j_3 \prec_\mathcal{I} j_2$.  This algorithm has a time complexity of $O(p^3)$, but if the number of parameters is large enough that this is too slow, the graph will likely be difficult to interpret or even display on a single page.

While the intervals and the cover graph convey the same information, they emphasize different aspects of it.  In most situations the best choice is to display both plots.  Figure~\ref{fig_viz} displays the cover graph for the order generated by the intervals shown in Figure~\ref{fig_simple}.

In~\cite{klein2020jointCR}, the authors suggest plotting the index intervals to visualize the uncertainty in a ranking.  We know that these can contain rankings incompatible with the data, but in addition it is difficult to determine which intervals overlap from the index intervals.  Consider the index intervals corresponding to the intervals in Figure~\ref{fig_counter}.  Given only these intervals, it is not immediately obvious that $1 \prec_I 4$.  As such, we do not recommend plotting the index intervals.

\section{Experiments with Simulated Data}
\label{sec_sim}

We now turn our attention to the operation of our algorithms on data.  In this section we examine the behavior of our methods in the context of normal mean estimation with an emphasis on studying $|\LE{\preceq_\mathcal{I}}|$ and the coverage probability.

For our experiments we take data from a five-dimensional multivariate Gaussian with independent components.  The means and standard deviations for each experiment are those found in~\cite{klein2011rankingNormal}:
\begin{enumerate}[(i)]

\item \label{exp_case_1} $\bm{\mu} = (10.0, 10.2, 10.4, 10.6, 10.8)$ and $\bm{\sigma} = (0.07, 0.07, 0.07, 0.07, 0.07)$;

\item \label{exp_case_2} $\bm{\mu} = (10.0, 10.2, 10.4, 10.6, 10.8)$ and $\bm{\sigma} = (0.05, 0.05, 0.2, 0.2, 0.2)$;

\item \label{exp_case_3} $\bm{\mu} = (10.0, 10.2, 10.7, 11.2, 11.4)$ and $\bm{\sigma} = (0.15, 0.15, 0.25, 0.15, 0.15)$;

\item \label{exp_case_4} $\bm{\mu} = (10.0, 10.5, 10.7, 11.0, 11.2)$ and $\bm{\sigma} = (0.1, 0.3, 0.3, 0.1, 0.5)$;

\item \label{exp_case_5} $\bm{\mu} = (9.8, 10.5, 10.7, 10.9, 11.6)$ and $\bm{\sigma} = (0.5, 0.1, 0.1, 0.1, 0.5)$.

\end{enumerate}

\begin{table}[p]
\begin{subtable}[h]{0.45\textwidth}
\centering
\begin{tabular}{|c|cccc|}
\hline
Case&$q(0.1)$&$q(0.5)$&$q(0.9)$&$p_c$\\
\hline
1 & 1 & 1 & 1 & 1 \\ 
2 & 1 & 2 & 2 & 1 \\ 
3 & 1 & 1 & 2 & 1 \\ 
4 & 2 & 4 & 4 & 1 \\ 
5 & 1 & 1 & 1 & 1 \\ 
\hline
\end{tabular}
\caption{$90\%$ Bonferroni confidence intervals.}
\label{tab_exp1-a}
\end{subtable}
\hfill
\begin{subtable}[h]{0.45\textwidth}
\centering
\begin{tabular}{|c|cccc|}
\hline
Case&$q(0.1)$&$q(0.5)$&$q(0.9)$&$p_c$\\
\hline
1 & 1 & 1 & 1 & 1 \\ 
2 & 1 & 1 & 2 & 1 \\ 
3 & 1 & 1 & 1 & 1 \\ 
4 & 1 & 2 & 4 & 1 \\ 
5 & 1 & 1 & 1 & 1 \\ 
\hline
\end{tabular}
\caption{$90\%$ unadjusted confidence intervals.}
\label{tab_exp1-b}
\end{subtable}

\begin{subtable}[h]{0.45\textwidth}
\centering
\begin{tabular}{|c|cccc|}
\hline
Case&$q(0.1)$&$q(0.5)$&$q(0.9)$&$p_c$\\
\hline
1 & 1 & 1 & 1 & 1.00 \\ 
2 & 1 & 2 & 3 & 1.00 \\ 
3 & 1 & 1 & 2 & 1.00 \\ 
4 & 2 & 4 & 5 & 1.00 \\ 
5 & 1 & 1 & 1 & 1.00 \\
\hline
\end{tabular}
\caption{$95\%$ Bonferroni confidence intervals.}
\label{tab_exp1-c}
\end{subtable}
\hfill
\begin{subtable}[h]{0.45\textwidth}
\centering
\begin{tabular}{|c|cccc|}
\hline
Case&$q(0.1)$&$q(0.5)$&$q(0.9)$&$p_c$\\
\hline
1 & 1 & 1 & 1 & 1 \\ 
2 & 1 & 1 & 2 & 1 \\ 
3 & 1 & 1 & 1 & 1 \\ 
4 & 1 & 2 & 4 & 1 \\ 
5 & 1 & 1 & 1 & 1 \\
\hline
\end{tabular}
\caption{$95\%$ unadjusted confidence intervals.}
\label{tab_exp1-d}
\end{subtable}
\caption{Results of the first normal mean estimation experiment.}
\label{tab_exp1}
\end{table}

\begin{table}[p]
\centering
\begin{tabular}{|c|cccc|}
\hline
Scaling Factor&$q(0.1)$&$q(0.5)$&$q(0.9)$&$p_c$\\
\hline
0.25 & 1  & 1  & 1  & 1 \\ 
0.5  & 1  & 1  & 1  & 1 \\ 
1    & 1  & 2  & 4  & 1 \\ 
2    & 4  & 8  & 12 & 1 \\ 
4    & 14 & 24 & 40 & 1 \\
\hline
\end{tabular}
\caption{Results of the second normal mean estimation experiment.}
\label{tab_exp2}
\end{table}

\begin{table}[p]
\centering
\begin{tabular}{|c|cc|}
\hline
$1 - \alpha$&$p_{\LE{\preceq_\mathcal{I}}}$&$p_{\text{pc}}$\\
\hline
0.90 & 1.00 & 1.00 \\ 
0.80 & 0.99 & 0.99 \\ 
0.70 & 0.98 & 0.98 \\ 
0.60 & 0.96 & 0.96 \\ 
0.50 & 0.89 & 0.89 \\ 
0.40 & 0.76 & 0.76 \\ 
0.30 & 0.60 & 0.60 \\ 
0.20 & 0.37 & 0.37 \\ 
0.10 & 0.18 & 0.18 \\
\hline
\end{tabular}
\caption{Results of the third normal mean estimation experiment.}
\label{tab_exp3}
\end{table}

For each of these cases, we use four confidence interval procedures based on combinations of a $90\%$ or $95\%$ confidence level and Bonferroni or unadjusted confidence intervals.  No multiple comparison adjustment will produce intervals which are wider than the Bonferroni intervals or narrower than the unadjusted intervals, so these represent the extreme possibilities.

For each trial, we draw $30$ observations from each distribution and examine the joint uncertainty in the ranks of the means based on the appropriate confidence intervals.  We perform $1000$ trials per combination of parameter values and confidence interval procedure.  Because the number of possible rankings is small, we compute exact counts of the rankings compatible with the data.

Table~\ref{tab_exp1} shows the results of our experiments.  In the table header $q(p)$ denotes the $p$th quantile of the number of compatible rankings and $p_c$ denotes the coverage probability.  We make the following observations about the results of our experiments:
\begin{enumerate}
\item As expected, the number of rankings compatible with the data increases with the width of the confidence intervals and the degree to which the densities overlap.
\item Even with a relatively small number of observations, the number of rankings compatible with the data is small.  Out of twenty experiments, none show $q(0.5) \geq 6$.
\item The coverage probability is exactly $1$ in all cases.  We investigate this in further experiments.
\end{enumerate}

In order to investigate the relationship between the density overlap and the number of rankings compatible with the data, we repeated the simulations for case~(\ref{exp_case_4}) with $95\%$ unadjusted confidence intervals and varying scaling factors applied to the standard deviations.  Table~\ref{tab_exp2} shows the results.  We note that even when the standard deviations are high, the number of rankings compatible with the data is still not usually above $20\%$ of all possible rankings.

Finally we investigate the coverage probability of $\LE{\preceq_I}$ and compare it against that of the product confidence set.  We simulate a 1000-dimensional normal distribution with $\mu_i = i$, $\sigma^2_i = 2$ and independent components.  We perform $1000$ trials for each confidence level shown in Table~\ref{tab_exp3}.  For each trial we draw $30$ observations, generate unadjusted confidence intervals, test whether the true ranking is compatible with the data, and test whether the true ranking is contained in the product confidence set.

Table~\ref{tab_exp3} shows the results of this experiment.  $p_{\LE{\preceq_\mathcal{I}}}$ denotes the coverage probability of $\LE{\preceq_\mathcal{I}}$, and $p_{\text{pc}}$ denotes the coverage probability of the product confidence set.  In this case, both methods are similarly conservative.  Further investigation is needed to understand the degree to which the product confidence set is more conservative than $\LE{\preceq_\mathcal{I}}$.

\section{A Case Study: Travel Time Data from the American Community Survey}
\label{sec_data}

In this section, we apply our methods to travel time data from the 2011 American Community Survey.  This data was previously considered in~\cite{klein2020jointCR}, where the authors computed the product confidence region and described it in terms of the index intervals.  Here we approximately measure $|\LE{\preceq_\mathcal{I}}|/p!$ for this data set.

\subsection{The Data}
\label{sec_data_desc}

\begin{figure}[p]
\centering
\input{tex/pstates_int.tex}
\caption{Interval plot for the travel time data.}
\label{fig_states_int}
\end{figure}
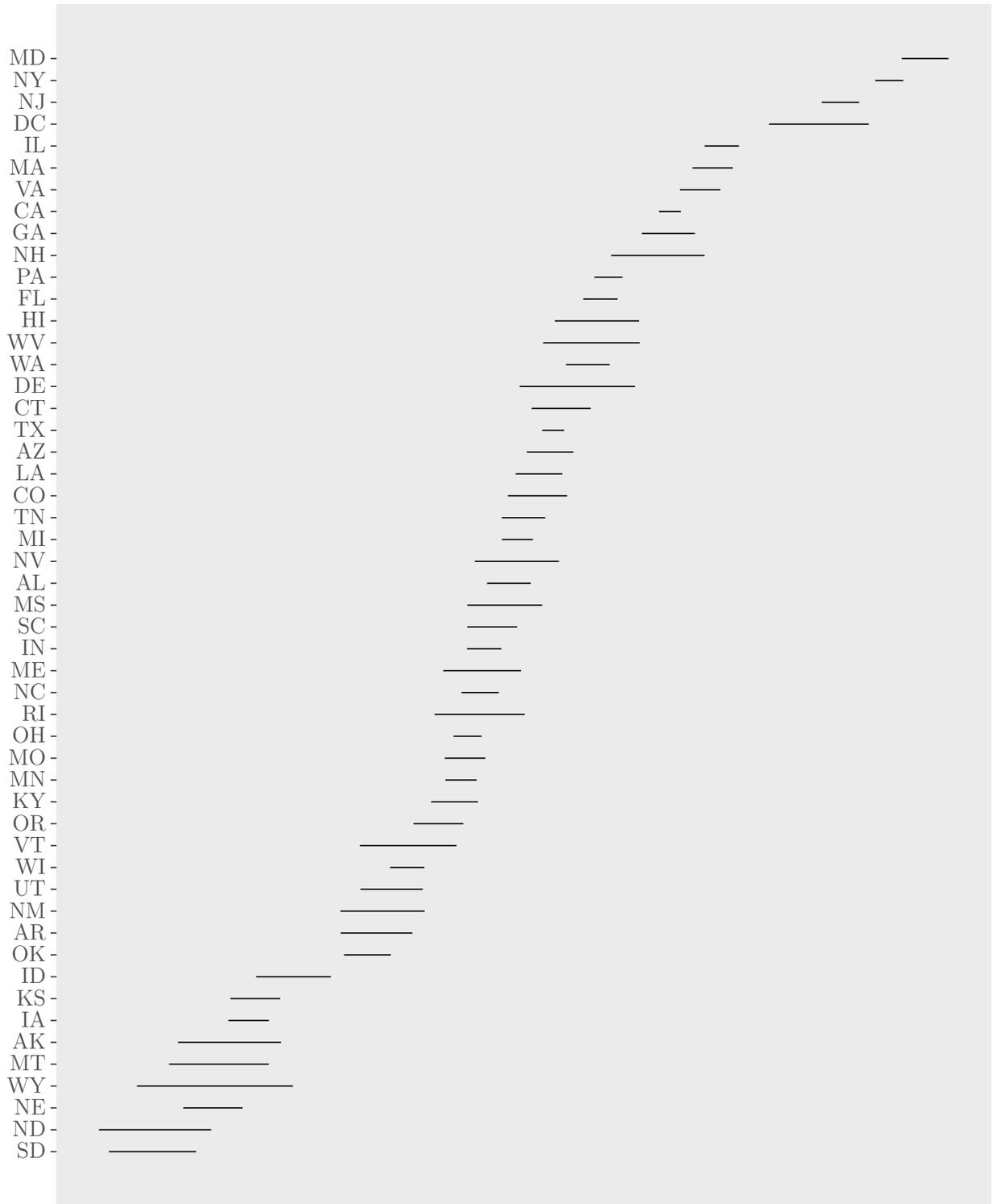

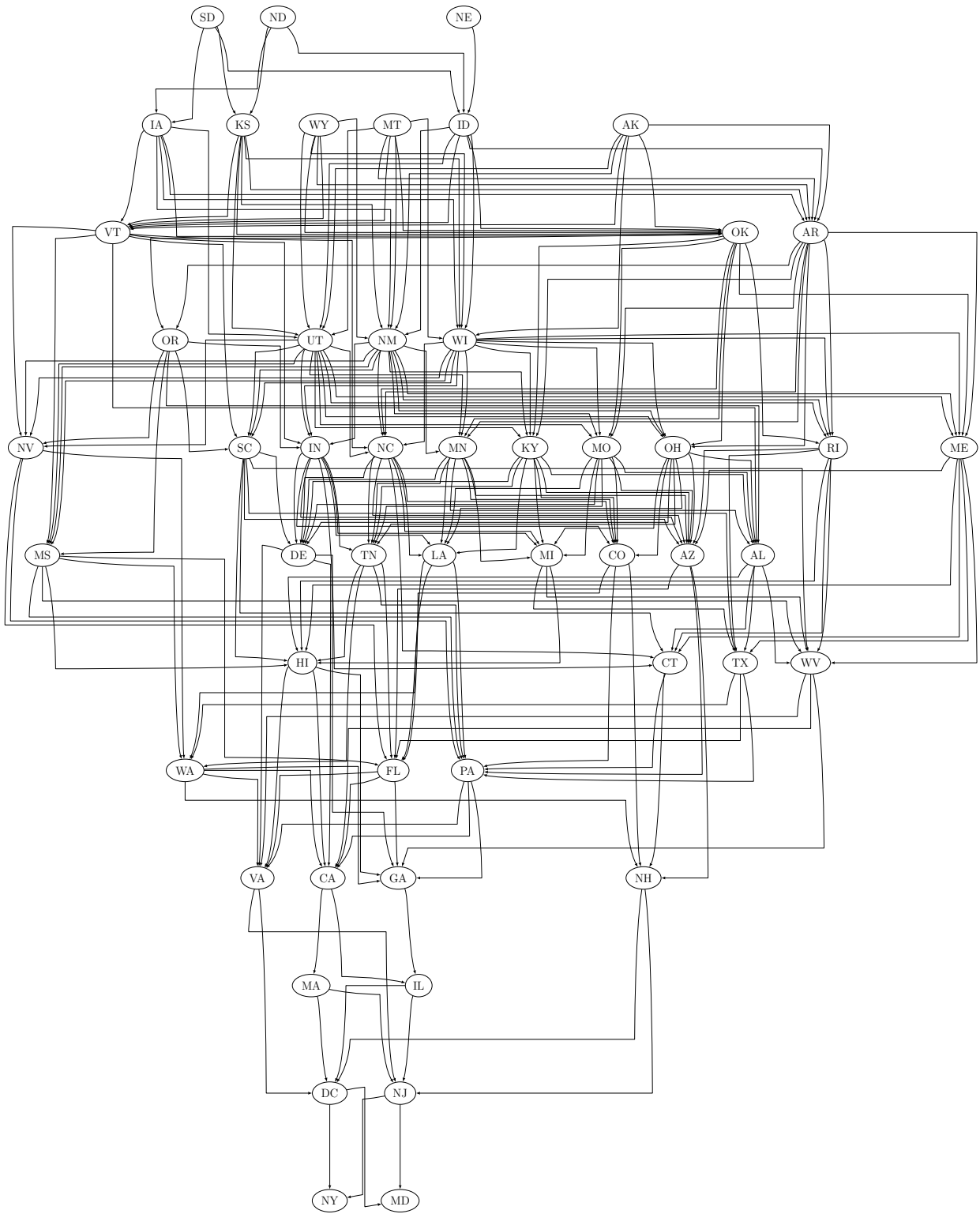
\begin{figure}[p]
\centering
\resizebox{6.5in}{8in}{%
\input{tex/pstates_graph.tex}
}
\caption{Cover graph for the travel time data.}
\label{fig_states_graph}
\end{figure}

The data for this case study is taken from the \texttt{TravelTime2011} dataset in the package \package{RankingProject}{wieczorek2020rankPackage}.  The data set describes average travel time to work for Americans over 16 years old who do not work from home in each US state and the District of Columbia.

We repeat the analysis in~\cite{klein2020jointCR} by fitting $90\%$ Bonferroni confidence intervals for the mean travel time.  Figures~\ref{fig_states_int} and~\ref{fig_states_graph} display the interval plot and cover graph for the order generated by the confidence intervals for the mean travel times.  We omit the numerical values of the confidence intervals in order to emphasize their order.  In Figure~\ref{fig_states_int}, it is clear that $\preceq_\mathcal{I}$ may be partitioned into four groups.  While Figure~\ref{fig_states_graph} does not make this so clear, the vertical positions of nodes corresponding to states conveys information about which pairs are comparable that is less clear from the interval plot.

\subsection{The Algorithm}

We use the Adaptive Relaxation Monte Carlo algorithm of~\cite{talvitie2018countInPractice}.  ARMC computes a randomized approximation of the number of linear extensions of a partial order controlled by parameters $\epsilon$ and $\delta$.  It guarantees that the approximation is within a factor of $1 + \epsilon$ of the true value with probability at least $1 - \delta$.  While its worst case running time is exponential in the size of the input, it often outperforms other asymptotically faster algorithms on small posets.

\subsection{Results}

\begin{table}[t]
\centering
\begin{tabular}{|c|c|c|c|}
\hline
Group&States&Count&Proportion\\
\hline
\ref{grp_one}&9&24192&0.07\\
\ref{grp_two}&37&$1.07 \times 10^{22}$&$7.77 \times 10^{-22}$\\
\ref{grp_three}&2&2&1\\
\ref{grp_four}&2&2&1\\
\hline
Overall&51&$1.03 \times 10^{27}$&$6.65 \times 10^{-40}$\\
\hline
\end{tabular}
\caption{Counts and proportions of rankings compatible with the data.}
\label{tab_counts}
\end{table}

As shown in Figure~\ref{fig_states_int}, the states can be partitioned into four groups.  We begin by listing these groups:

\begin{enumerate}

\item\label{grp_one} AK, IA, ID, KS, MT, ND, NE, SD, WY

\item\label{grp_two} AL, AR, AZ, CA, CO, CT, DE, FL, GA, HI, IL, IN, KY, LA, MA, ME, MI, MN, MO, MS, NC, NH, NM, NV, OH, OK, OR, PA, RI, SC, TN, TX, UT, VA, VT, WA, WI, WV

\item\label{grp_three} DC, NJ

\item\label{grp_four} MD, NY

\end{enumerate}

Table~\ref{tab_counts} shows the number and proportion of rankings compatible with the data for each group as well as all the states together.  The counts for groups \ref{grp_one}, \ref{grp_three} and~\ref{grp_four} are exact.  The count for group \ref{grp_two} is generated by ARMC with $\epsilon = 0.01$ and $\delta = 0.005$.  The proportions for each group are calculated by dividing the number of rankings compatible with the data in each partition by the number of possible rankings of that partition.

These experiments were run on a PC with a 2.8 GHz Intel Xeon CPU.  With these parameters, approximating the number of linear extensions of group \ref{grp_one} takes roughly one second.  Performing the same approximation for all states takes only about $30\%$ longer.

\section*{Acknowledgments}

We have used \package{R}{r} and the following packages for our experiments: \linebreak\package{arrangements}{lai2020arrangements}, \package{ggplot2}{wickham2016ggplot2}, \package{magrittr}{bache2014magrittr}, \package{Matrix}{bates2020matrix}, \package{mvtnorm}{genz2020mvtnorm}, \package{reshape2}{wickham2007reshape}, \package{Rcpp}{eddelbuettel2011rcpp}, \package{tikzDevice}{sharpsteen2020tikzDevice}, \linebreak\package{xtable}{dahl2019xtable}.  We have also used code provided by the authors of~\cite{talvitie2018countInPractice} for the analysis in section~\ref{sec_data} and \package{Graphviz}{gansner2000graphviz} for Figures~\ref{fig_viz} and~\ref{fig_states_graph}.

\section*{Supplemental Materials}
\label{sec_supp}

The methods described in this paper are implemented in the \texttt{R} package \texttt{rankUncertainty}.

\bibliographystyle{apalike}
\bibliography{UncertaintyInRanking}

\end{document}

%% file: tex/pex.tex
\begin{tikzpicture}[x=1pt,y=1pt]
\definecolor{fillColor}{RGB}{255,255,255}
\path[use as bounding box,fill=fillColor,fill opacity=0.00] (0,0) rectangle (289.08,144.54);
\begin{scope}
\path[clip] (  0.00,  0.00) rectangle (289.08,144.54);
\definecolor{drawColor}{RGB}{255,255,255}
\definecolor{fillColor}{RGB}{255,255,255}

\path[draw=drawColor,line width= 0.6pt,line join=round,line cap=round,fill=fillColor] (  0.00,  0.00) rectangle (289.08,144.54);
\end{scope}
\begin{scope}
\path[clip] (  8.25,  8.25) rectangle (283.58,139.04);
\definecolor{fillColor}{gray}{0.92}

\path[fill=fillColor] (  8.25,  8.25) rectangle (283.58,139.04);
\definecolor{drawColor}{RGB}{0,0,0}

\path[draw=drawColor,line width= 0.6pt,line join=round] ( 20.77, 14.20) -- (128.04, 14.20);

\path[draw=drawColor,line width= 0.6pt,line join=round] ( 92.28, 73.64) -- (199.55, 73.64);

\path[draw=drawColor,line width= 0.6pt,line join=round] (163.79,133.10) -- (271.06,133.10);
\end{scope}
\end{tikzpicture}

%% file: tex/pfreq.tex
\begin{tikzpicture}[x=1pt,y=1pt]
\definecolor{fillColor}{RGB}{255,255,255}
\path[use as bounding box,fill=fillColor,fill opacity=0.00] (0,0) rectangle (289.08,144.54);
\begin{scope}
\path[clip] (  0.00,  0.00) rectangle (289.08,144.54);
\definecolor{drawColor}{RGB}{255,255,255}
\definecolor{fillColor}{RGB}{255,255,255}

\path[draw=drawColor,line width= 0.6pt,line join=round,line cap=round,fill=fillColor] (  0.00,  0.00) rectangle (289.08,144.54);
\end{scope}
\begin{scope}
\path[clip] (  8.25,  8.25) rectangle (283.58,139.04);
\definecolor{fillColor}{gray}{0.92}

\path[fill=fillColor] (  8.25,  8.25) rectangle (283.58,139.04);
\definecolor{drawColor}{RGB}{0,0,0}
\definecolor{fillColor}{RGB}{0,0,0}

\path[draw=drawColor,line width= 0.4pt,line join=round,line cap=round,fill=fillColor] ( 80.36, 14.20) circle (  1.96);

\path[draw=drawColor,line width= 0.4pt,line join=round,line cap=round,fill=fillColor] (104.20, 73.64) circle (  1.96);

\path[draw=drawColor,line width= 0.4pt,line join=round,line cap=round,fill=fillColor] (271.06,133.10) circle (  1.96);

\path[draw=drawColor,line width= 0.6pt,line join=round] ( 20.77, 14.20) -- ( 68.44, 14.20);

\path[draw=drawColor,line width= 0.6pt,line join=round] (116.12, 73.64) -- (163.79, 73.64);

\path[draw=drawColor,line width= 0.6pt,line join=round] (211.47,133.10) -- (259.15,133.10);
\end{scope}
\end{tikzpicture}

%% file: tex/pcounter.tex
\begin{tikzpicture}[x=1pt,y=1pt]
\definecolor{fillColor}{RGB}{255,255,255}
\path[use as bounding box,fill=fillColor,fill opacity=0.00] (0,0) rectangle (289.08,144.54);
\begin{scope}
\path[clip] (  0.00,  0.00) rectangle (289.08,144.54);
\definecolor{drawColor}{RGB}{255,255,255}
\definecolor{fillColor}{RGB}{255,255,255}

\path[draw=drawColor,line width= 0.6pt,line join=round,line cap=round,fill=fillColor] (  0.00,  0.00) rectangle (289.08,144.54);
\end{scope}
\begin{scope}
\path[clip] (  8.25,  8.25) rectangle (283.58,139.04);
\definecolor{fillColor}{gray}{0.92}

\path[fill=fillColor] (  8.25,  8.25) rectangle (283.58,139.04);
\definecolor{drawColor}{RGB}{0,0,0}

\path[draw=drawColor,line width= 0.6pt,line join=round] ( 20.77, 14.20) -- (125.96, 14.20);

\path[draw=drawColor,line width= 0.6pt,line join=round] ( 57.04, 43.92) -- (162.24, 43.92);

\path[draw=drawColor,line width= 0.6pt,line join=round] ( 93.32, 73.64) -- (198.51, 73.64);

\path[draw=drawColor,line width= 0.6pt,line join=round] (129.59,103.37) -- (234.79,103.37);

\path[draw=drawColor,line width= 0.6pt,line join=round] (165.87,133.10) -- (271.06,133.10);
\end{scope}
\end{tikzpicture}

%% file: tex/pcover.tex
\begin{tikzpicture}[>=latex,line join=bevel,]
\node (a) at (46.146bp,46.146bp) [draw,ellipse] {$\theta_1$};
  \node (c) at (174.44bp,46.146bp) [draw,ellipse] {$\theta_3$};
  \node (b) at (110.29bp,156.15bp) [draw,ellipse] {$\theta_2$}; 
  \draw [->] (a) ..controls (100.73bp,46.146bp) and (109.42bp,46.146bp)  .. (c);
\end{tikzpicture}

%% file: tex/pstates_int.tex
\begin{tikzpicture}[x=1pt,y=1pt]
\definecolor{fillColor}{RGB}{255,255,255}
\path[use as bounding box,fill=fillColor,fill opacity=0.00] (0,0) rectangle (469.76,578.16);
\begin{scope}
\path[clip] (  0.00,  0.00) rectangle (469.76,578.16);
\definecolor{drawColor}{RGB}{255,255,255}
\definecolor{fillColor}{RGB}{255,255,255}

\path[draw=drawColor,line width= 0.6pt,line join=round,line cap=round,fill=fillColor] (  0.00,  0.00) rectangle (469.76,578.16);
\end{scope}
\begin{scope}
\path[clip] ( 26.09,  8.25) rectangle (464.25,572.66);
\definecolor{fillColor}{gray}{0.92}

\path[fill=fillColor] ( 26.09,  8.25) rectangle (464.26,572.66);
\definecolor{drawColor}{RGB}{0,0,0}

\path[draw=drawColor,line width= 0.6pt,line join=round] ( 50.67, 33.91) -- ( 91.59, 33.91);

\path[draw=drawColor,line width= 0.6pt,line join=round] ( 46.01, 44.17) -- ( 98.61, 44.17);

\path[draw=drawColor,line width= 0.6pt,line join=round] ( 85.57, 54.43) -- (113.33, 54.43);

\path[draw=drawColor,line width= 0.6pt,line join=round] ( 63.86, 64.69) -- (136.93, 64.69);

\path[draw=drawColor,line width= 0.6pt,line join=round] ( 78.90, 74.95) -- (125.66, 74.95);

\path[draw=drawColor,line width= 0.6pt,line join=round] ( 83.13, 85.22) -- (131.35, 85.22);

\path[draw=drawColor,line width= 0.6pt,line join=round] (106.71, 95.48) -- (125.70, 95.48);

\path[draw=drawColor,line width= 0.6pt,line join=round] (107.58,105.74) -- (130.96,105.74);

\path[draw=drawColor,line width= 0.6pt,line join=round] (119.67,116.00) -- (154.74,116.00);

\path[draw=drawColor,line width= 0.6pt,line join=round] (160.94,126.26) -- (182.86,126.26);

\path[draw=drawColor,line width= 0.6pt,line join=round] (159.34,136.53) -- (192.95,136.53);

\path[draw=drawColor,line width= 0.6pt,line join=round] (159.25,146.79) -- (198.71,146.79);

\path[draw=drawColor,line width= 0.6pt,line join=round] (168.61,157.05) -- (197.84,157.05);

\path[draw=drawColor,line width= 0.6pt,line join=round] (182.51,167.31) -- (198.58,167.31);

\path[draw=drawColor,line width= 0.6pt,line join=round] (168.36,177.57) -- (213.67,177.57);

\path[draw=drawColor,line width= 0.6pt,line join=round] (193.48,187.84) -- (216.86,187.84);

\path[draw=drawColor,line width= 0.6pt,line join=round] (201.77,198.10) -- (223.69,198.10);

\path[draw=drawColor,line width= 0.6pt,line join=round] (208.49,208.36) -- (223.10,208.36);

\path[draw=drawColor,line width= 0.6pt,line join=round] (208.18,218.62) -- (227.18,218.62);

\path[draw=drawColor,line width= 0.6pt,line join=round] (212.29,228.88) -- (225.44,228.88);

\path[draw=drawColor,line width= 0.6pt,line join=round] (203.34,239.15) -- (245.71,239.15);

\path[draw=drawColor,line width= 0.6pt,line join=round] (215.99,249.41) -- (233.53,249.41);

\path[draw=drawColor,line width= 0.6pt,line join=round] (207.44,259.67) -- (243.97,259.67);

\path[draw=drawColor,line width= 0.6pt,line join=round] (218.61,269.93) -- (234.69,269.93);

\path[draw=drawColor,line width= 0.6pt,line join=round] (218.74,280.19) -- (242.12,280.19);

\path[draw=drawColor,line width= 0.6pt,line join=round] (218.79,290.46) -- (253.86,290.46);

\path[draw=drawColor,line width= 0.6pt,line join=round] (227.98,300.72) -- (248.44,300.72);

\path[draw=drawColor,line width= 0.6pt,line join=round] (222.26,310.98) -- (261.72,310.98);

\path[draw=drawColor,line width= 0.6pt,line join=round] (234.92,321.24) -- (249.53,321.24);

\path[draw=drawColor,line width= 0.6pt,line join=round] (234.83,331.50) -- (255.29,331.50);

\path[draw=drawColor,line width= 0.6pt,line join=round] (237.78,341.76) -- (265.55,341.76);

\path[draw=drawColor,line width= 0.6pt,line join=round] (241.41,352.03) -- (263.33,352.03);

\path[draw=drawColor,line width= 0.6pt,line join=round] (246.60,362.29) -- (268.52,362.29);

\path[draw=drawColor,line width= 0.6pt,line join=round] (253.87,372.55) -- (264.10,372.55);

\path[draw=drawColor,line width= 0.6pt,line join=round] (248.87,382.81) -- (276.64,382.81);

\path[draw=drawColor,line width= 0.6pt,line join=round] (243.27,393.07) -- (297.34,393.07);

\path[draw=drawColor,line width= 0.6pt,line join=round] (265.03,403.34) -- (285.49,403.34);

\path[draw=drawColor,line width= 0.6pt,line join=round] (254.27,413.60) -- (299.57,413.60);

\path[draw=drawColor,line width= 0.6pt,line join=round] (259.78,423.86) -- (299.24,423.86);

\path[draw=drawColor,line width= 0.6pt,line join=round] (273.13,434.12) -- (289.20,434.12);

\path[draw=drawColor,line width= 0.6pt,line join=round] (278.36,444.39) -- (291.52,444.39);

\path[draw=drawColor,line width= 0.6pt,line join=round] (286.15,454.65) -- (329.99,454.65);

\path[draw=drawColor,line width= 0.6pt,line join=round] (300.60,464.91) -- (325.44,464.91);

\path[draw=drawColor,line width= 0.6pt,line join=round] (308.62,475.17) -- (318.85,475.17);

\path[draw=drawColor,line width= 0.6pt,line join=round] (318.39,485.43) -- (337.39,485.43);

\path[draw=drawColor,line width= 0.6pt,line join=round] (324.29,495.69) -- (343.29,495.69);

\path[draw=drawColor,line width= 0.6pt,line join=round] (330.00,505.96) -- (346.08,505.96);

\path[draw=drawColor,line width= 0.6pt,line join=round] (360.20,516.22) -- (406.97,516.22);

\path[draw=drawColor,line width= 0.6pt,line join=round] (384.96,526.48) -- (402.50,526.48);

\path[draw=drawColor,line width= 0.6pt,line join=round] (410.05,536.74) -- (423.20,536.74);

\path[draw=drawColor,line width= 0.6pt,line join=round] (422.42,547.01) -- (444.34,547.01);
\end{scope}
\begin{scope}
\path[clip] (  0.00,  0.00) rectangle (469.76,578.16);
\definecolor{drawColor}{gray}{0.30}

\node[text=drawColor,anchor=base east,inner sep=0pt, outer sep=0pt, scale=  0.88] at ( 21.14, 30.87) {SD};

\node[text=drawColor,anchor=base east,inner sep=0pt, outer sep=0pt, scale=  0.88] at ( 21.14, 41.14) {ND};

\node[text=drawColor,anchor=base east,inner sep=0pt, outer sep=0pt, scale=  0.88] at ( 21.14, 51.40) {NE};

\node[text=drawColor,anchor=base east,inner sep=0pt, outer sep=0pt, scale=  0.88] at ( 21.14, 61.66) {WY};

\node[text=drawColor,anchor=base east,inner sep=0pt, outer sep=0pt, scale=  0.88] at ( 21.14, 71.92) {MT};

\node[text=drawColor,anchor=base east,inner sep=0pt, outer sep=0pt, scale=  0.88] at ( 21.14, 82.18) {AK};

\node[text=drawColor,anchor=base east,inner sep=0pt, outer sep=0pt, scale=  0.88] at ( 21.14, 92.45) {IA};

\node[text=drawColor,anchor=base east,inner sep=0pt, outer sep=0pt, scale=  0.88] at ( 21.14,102.71) {KS};

\node[text=drawColor,anchor=base east,inner sep=0pt, outer sep=0pt, scale=  0.88] at ( 21.14,112.97) {ID};

\node[text=drawColor,anchor=base east,inner sep=0pt, outer sep=0pt, scale=  0.88] at ( 21.14,123.23) {OK};

\node[text=drawColor,anchor=base east,inner sep=0pt, outer sep=0pt, scale=  0.88] at ( 21.14,133.49) {AR};

\node[text=drawColor,anchor=base east,inner sep=0pt, outer sep=0pt, scale=  0.88] at ( 21.14,143.76) {NM};

\node[text=drawColor,anchor=base east,inner sep=0pt, outer sep=0pt, scale=  0.88] at ( 21.14,154.02) {UT};

\node[text=drawColor,anchor=base east,inner sep=0pt, outer sep=0pt, scale=  0.88] at ( 21.14,164.28) {WI};

\node[text=drawColor,anchor=base east,inner sep=0pt, outer sep=0pt, scale=  0.88] at ( 21.14,174.54) {VT};

\node[text=drawColor,anchor=base east,inner sep=0pt, outer sep=0pt, scale=  0.88] at ( 21.14,184.80) {OR};

\node[text=drawColor,anchor=base east,inner sep=0pt, outer sep=0pt, scale=  0.88] at ( 21.14,195.07) {KY};

\node[text=drawColor,anchor=base east,inner sep=0pt, outer sep=0pt, scale=  0.88] at ( 21.14,205.33) {MN};

\node[text=drawColor,anchor=base east,inner sep=0pt, outer sep=0pt, scale=  0.88] at ( 21.14,215.59) {MO};

\node[text=drawColor,anchor=base east,inner sep=0pt, outer sep=0pt, scale=  0.88] at ( 21.14,225.85) {OH};

\node[text=drawColor,anchor=base east,inner sep=0pt, outer sep=0pt, scale=  0.88] at ( 21.14,236.11) {RI};

\node[text=drawColor,anchor=base east,inner sep=0pt, outer sep=0pt, scale=  0.88] at ( 21.14,246.38) {NC};

\node[text=drawColor,anchor=base east,inner sep=0pt, outer sep=0pt, scale=  0.88] at ( 21.14,256.64) {ME};

\node[text=drawColor,anchor=base east,inner sep=0pt, outer sep=0pt, scale=  0.88] at ( 21.14,266.90) {IN};

\node[text=drawColor,anchor=base east,inner sep=0pt, outer sep=0pt, scale=  0.88] at ( 21.14,277.16) {SC};

\node[text=drawColor,anchor=base east,inner sep=0pt, outer sep=0pt, scale=  0.88] at ( 21.14,287.42) {MS};

\node[text=drawColor,anchor=base east,inner sep=0pt, outer sep=0pt, scale=  0.88] at ( 21.14,297.69) {AL};

\node[text=drawColor,anchor=base east,inner sep=0pt, outer sep=0pt, scale=  0.88] at ( 21.14,307.95) {NV};

\node[text=drawColor,anchor=base east,inner sep=0pt, outer sep=0pt, scale=  0.88] at ( 21.14,318.21) {MI};

\node[text=drawColor,anchor=base east,inner sep=0pt, outer sep=0pt, scale=  0.88] at ( 21.14,328.47) {TN};

\node[text=drawColor,anchor=base east,inner sep=0pt, outer sep=0pt, scale=  0.88] at ( 21.14,338.73) {CO};

\node[text=drawColor,anchor=base east,inner sep=0pt, outer sep=0pt, scale=  0.88] at ( 21.14,349.00) {LA};

\node[text=drawColor,anchor=base east,inner sep=0pt, outer sep=0pt, scale=  0.88] at ( 21.14,359.26) {AZ};

\node[text=drawColor,anchor=base east,inner sep=0pt, outer sep=0pt, scale=  0.88] at ( 21.14,369.52) {TX};

\node[text=drawColor,anchor=base east,inner sep=0pt, outer sep=0pt, scale=  0.88] at ( 21.14,379.78) {CT};

\node[text=drawColor,anchor=base east,inner sep=0pt, outer sep=0pt, scale=  0.88] at ( 21.14,390.04) {DE};

\node[text=drawColor,anchor=base east,inner sep=0pt, outer sep=0pt, scale=  0.88] at ( 21.14,400.31) {WA};

\node[text=drawColor,anchor=base east,inner sep=0pt, outer sep=0pt, scale=  0.88] at ( 21.14,410.57) {WV};

\node[text=drawColor,anchor=base east,inner sep=0pt, outer sep=0pt, scale=  0.88] at ( 21.14,420.83) {HI};

\node[text=drawColor,anchor=base east,inner sep=0pt, outer sep=0pt, scale=  0.88] at ( 21.14,431.09) {FL};

\node[text=drawColor,anchor=base east,inner sep=0pt, outer sep=0pt, scale=  0.88] at ( 21.14,441.35) {PA};

\node[text=drawColor,anchor=base east,inner sep=0pt, outer sep=0pt, scale=  0.88] at ( 21.14,451.62) {NH};

\node[text=drawColor,anchor=base east,inner sep=0pt, outer sep=0pt, scale=  0.88] at ( 21.14,461.88) {GA};

\node[text=drawColor,anchor=base east,inner sep=0pt, outer sep=0pt, scale=  0.88] at ( 21.14,472.14) {CA};

\node[text=drawColor,anchor=base east,inner sep=0pt, outer sep=0pt, scale=  0.88] at ( 21.14,482.40) {VA};

\node[text=drawColor,anchor=base east,inner sep=0pt, outer sep=0pt, scale=  0.88] at ( 21.14,492.66) {MA};

\node[text=drawColor,anchor=base east,inner sep=0pt, outer sep=0pt, scale=  0.88] at ( 21.14,502.93) {IL};

\node[text=drawColor,anchor=base east,inner sep=0pt, outer sep=0pt, scale=  0.88] at ( 21.14,513.19) {DC};

\node[text=drawColor,anchor=base east,inner sep=0pt, outer sep=0pt, scale=  0.88] at ( 21.14,523.45) {NJ};

\node[text=drawColor,anchor=base east,inner sep=0pt, outer sep=0pt, scale=  0.88] at ( 21.14,533.71) {NY};

\node[text=drawColor,anchor=base east,inner sep=0pt, outer sep=0pt, scale=  0.88] at ( 21.14,543.97) {MD};
\end{scope}
\begin{scope}
\path[clip] (  0.00,  0.00) rectangle (469.76,578.16);
\definecolor{drawColor}{gray}{0.20}

\path[draw=drawColor,line width= 0.6pt,line join=round] ( 23.34, 33.91) --
	( 26.09, 33.91);

\path[draw=drawColor,line width= 0.6pt,line join=round] ( 23.34, 44.17) --
	( 26.09, 44.17);

\path[draw=drawColor,line width= 0.6pt,line join=round] ( 23.34, 54.43) --
	( 26.09, 54.43);

\path[draw=drawColor,line width= 0.6pt,line join=round] ( 23.34, 64.69) --
	( 26.09, 64.69);

\path[draw=drawColor,line width= 0.6pt,line join=round] ( 23.34, 74.95) --
	( 26.09, 74.95);

\path[draw=drawColor,line width= 0.6pt,line join=round] ( 23.34, 85.22) --
	( 26.09, 85.22);

\path[draw=drawColor,line width= 0.6pt,line join=round] ( 23.34, 95.48) --
	( 26.09, 95.48);

\path[draw=drawColor,line width= 0.6pt,line join=round] ( 23.34,105.74) --
	( 26.09,105.74);

\path[draw=drawColor,line width= 0.6pt,line join=round] ( 23.34,116.00) --
	( 26.09,116.00);

\path[draw=drawColor,line width= 0.6pt,line join=round] ( 23.34,126.26) --
	( 26.09,126.26);

\path[draw=drawColor,line width= 0.6pt,line join=round] ( 23.34,136.53) --
	( 26.09,136.53);

\path[draw=drawColor,line width= 0.6pt,line join=round] ( 23.34,146.79) --
	( 26.09,146.79);

\path[draw=drawColor,line width= 0.6pt,line join=round] ( 23.34,157.05) --
	( 26.09,157.05);

\path[draw=drawColor,line width= 0.6pt,line join=round] ( 23.34,167.31) --
	( 26.09,167.31);

\path[draw=drawColor,line width= 0.6pt,line join=round] ( 23.34,177.57) --
	( 26.09,177.57);

\path[draw=drawColor,line width= 0.6pt,line join=round] ( 23.34,187.84) --
	( 26.09,187.84);

\path[draw=drawColor,line width= 0.6pt,line join=round] ( 23.34,198.10) --
	( 26.09,198.10);

\path[draw=drawColor,line width= 0.6pt,line join=round] ( 23.34,208.36) --
	( 26.09,208.36);

\path[draw=drawColor,line width= 0.6pt,line join=round] ( 23.34,218.62) --
	( 26.09,218.62);

\path[draw=drawColor,line width= 0.6pt,line join=round] ( 23.34,228.88) --
	( 26.09,228.88);

\path[draw=drawColor,line width= 0.6pt,line join=round] ( 23.34,239.15) --
	( 26.09,239.15);

\path[draw=drawColor,line width= 0.6pt,line join=round] ( 23.34,249.41) --
	( 26.09,249.41);

\path[draw=drawColor,line width= 0.6pt,line join=round] ( 23.34,259.67) --
	( 26.09,259.67);

\path[draw=drawColor,line width= 0.6pt,line join=round] ( 23.34,269.93) --
	( 26.09,269.93);

\path[draw=drawColor,line width= 0.6pt,line join=round] ( 23.34,280.19) --
	( 26.09,280.19);

\path[draw=drawColor,line width= 0.6pt,line join=round] ( 23.34,290.46) --
	( 26.09,290.46);

\path[draw=drawColor,line width= 0.6pt,line join=round] ( 23.34,300.72) --
	( 26.09,300.72);

\path[draw=drawColor,line width= 0.6pt,line join=round] ( 23.34,310.98) --
	( 26.09,310.98);

\path[draw=drawColor,line width= 0.6pt,line join=round] ( 23.34,321.24) --
	( 26.09,321.24);

\path[draw=drawColor,line width= 0.6pt,line join=round] ( 23.34,331.50) --
	( 26.09,331.50);

\path[draw=drawColor,line width= 0.6pt,line join=round] ( 23.34,341.76) --
	( 26.09,341.76);

\path[draw=drawColor,line width= 0.6pt,line join=round] ( 23.34,352.03) --
	( 26.09,352.03);

\path[draw=drawColor,line width= 0.6pt,line join=round] ( 23.34,362.29) --
	( 26.09,362.29);

\path[draw=drawColor,line width= 0.6pt,line join=round] ( 23.34,372.55) --
	( 26.09,372.55);

\path[draw=drawColor,line width= 0.6pt,line join=round] ( 23.34,382.81) --
	( 26.09,382.81);

\path[draw=drawColor,line width= 0.6pt,line join=round] ( 23.34,393.07) --
	( 26.09,393.07);

\path[draw=drawColor,line width= 0.6pt,line join=round] ( 23.34,403.34) --
	( 26.09,403.34);

\path[draw=drawColor,line width= 0.6pt,line join=round] ( 23.34,413.60) --
	( 26.09,413.60);

\path[draw=drawColor,line width= 0.6pt,line join=round] ( 23.34,423.86) --
	( 26.09,423.86);

\path[draw=drawColor,line width= 0.6pt,line join=round] ( 23.34,434.12) --
	( 26.09,434.12);

\path[draw=drawColor,line width= 0.6pt,line join=round] ( 23.34,444.39) --
	( 26.09,444.39);

\path[draw=drawColor,line width= 0.6pt,line join=round] ( 23.34,454.65) --
	( 26.09,454.65);

\path[draw=drawColor,line width= 0.6pt,line join=round] ( 23.34,464.91) --
	( 26.09,464.91);

\path[draw=drawColor,line width= 0.6pt,line join=round] ( 23.34,475.17) --
	( 26.09,475.17);

\path[draw=drawColor,line width= 0.6pt,line join=round] ( 23.34,485.43) --
	( 26.09,485.43);

\path[draw=drawColor,line width= 0.6pt,line join=round] ( 23.34,495.69) --
	( 26.09,495.69);

\path[draw=drawColor,line width= 0.6pt,line join=round] ( 23.34,505.96) --
	( 26.09,505.96);

\path[draw=drawColor,line width= 0.6pt,line join=round] ( 23.34,516.22) --
	( 26.09,516.22);

\path[draw=drawColor,line width= 0.6pt,line join=round] ( 23.34,526.48) --
	( 26.09,526.48);

\path[draw=drawColor,line width= 0.6pt,line join=round] ( 23.34,536.74) --
	( 26.09,536.74);

\path[draw=drawColor,line width= 0.6pt,line join=round] ( 23.34,547.01) --
	( 26.09,547.01);
\end{scope}
\end{tikzpicture}

%% file: tex/pstates_graph.tex
\begin{tikzpicture}[>=latex,line join=bevel,]
\node (WA) at (190.0bp,459.0bp) [draw,ellipse] {WA};
  \node (DE) at (306.0bp,675.0bp) [draw,ellipse] {DE};
  \node (DC) at (338.0bp,135.0bp) [draw,ellipse] {DC};
  \node (WI) at (471.0bp,891.0bp) [draw,ellipse] {WI};
  \node (WV) at (830.0bp,567.0bp) [draw,ellipse] {WV};
  \node (HI) at (310.0bp,567.0bp) [draw,ellipse] {HI};
  \node (FL) at (403.0bp,459.0bp) [draw,ellipse] {FL};
  \node (WY) at (327.0bp,1107.0bp) [draw,ellipse] {WY};
  \node (NH) at (659.0bp,351.0bp) [draw,ellipse] {NH};
  \node (NJ) at (410.0bp,135.0bp) [draw,ellipse] {NJ};
  \node (NM) at (397.0bp,891.0bp) [draw,ellipse] {NM};
  \node (TX) at (758.0bp,567.0bp) [draw,ellipse] {TX};
  \node (LA) at (450.0bp,675.0bp) [draw,ellipse] {LA};
  \node (NC) at (395.0bp,783.0bp) [draw,ellipse] {NC};
  \node (ND) at (285.0bp,1215.0bp) [draw,ellipse] {ND};
  \node (NE) at (475.0bp,1215.0bp) [draw,ellipse] {NE};
  \node (TN) at (378.0bp,675.0bp) [draw,ellipse] {TN};
  \node (NY) at (338.0bp,27.0bp) [draw,ellipse] {NY};
  \node (PA) at (479.0bp,459.0bp) [draw,ellipse] {PA};
  \node (RI) at (853.0bp,783.0bp) [draw,ellipse] {RI};
  \node (NV) at (27.0bp,783.0bp) [draw,ellipse] {NV};
  \node (VA) at (264.0bp,351.0bp) [draw,ellipse] {VA};
  \node (CO) at (632.0bp,675.0bp) [draw,ellipse] {CO};
  \node (AK) at (646.0bp,1107.0bp) [draw,ellipse] {AK};
  \node (AL) at (776.0bp,675.0bp) [draw,ellipse] {AL};
  \node (AR) at (830.0bp,999.0bp) [draw,ellipse] {AR};
  \node (VT) at (116.0bp,999.0bp) [draw,ellipse] {VT};
  \node (IL) at (429.0bp,243.0bp) [draw,ellipse] {IL};
  \node (GA) at (408.0bp,351.0bp) [draw,ellipse] {GA};
  \node (IN) at (323.0bp,783.0bp) [draw,ellipse] {IN};
  \node (IA) at (161.0bp,1107.0bp) [draw,ellipse] {IA};
  \node (MA) at (319.0bp,243.0bp) [draw,ellipse] {MA};
  \node (AZ) at (704.0bp,675.0bp) [draw,ellipse] {AZ};
  \node (CA) at (336.0bp,351.0bp) [draw,ellipse] {CA};
  \node (ID) at (475.0bp,1107.0bp) [draw,ellipse] {ID};
  \node (CT) at (686.0bp,567.0bp) [draw,ellipse] {CT};
  \node (ME) at (982.0bp,783.0bp) [draw,ellipse] {ME};
  \node (MD) at (410.0bp,27.0bp) [draw,ellipse] {MD};
  \node (OK) at (758.0bp,999.0bp) [draw,ellipse] {OK};
  \node (OH) at (689.0bp,783.0bp) [draw,ellipse] {OH};
  \node (UT) at (323.0bp,891.0bp) [draw,ellipse] {UT};
  \node (MO) at (616.0bp,783.0bp) [draw,ellipse] {MO};
  \node (MN) at (469.0bp,783.0bp) [draw,ellipse] {MN};
  \node (MI) at (560.0bp,675.0bp) [draw,ellipse] {MI};
  \node (KS) at (249.0bp,1107.0bp) [draw,ellipse] {KS};
  \node (MT) at (402.0bp,1107.0bp) [draw,ellipse] {MT};
  \node (MS) at (44.0bp,675.0bp) [draw,ellipse] {MS};
  \node (SC) at (251.0bp,783.0bp) [draw,ellipse] {SC};
  \node (KY) at (543.0bp,783.0bp) [draw,ellipse] {KY};
  \node (OR) at (175.0bp,891.0bp) [draw,ellipse] {OR};
  \node (SD) at (213.0bp,1215.0bp) [draw,ellipse] {SD};
  \draw [->] (IN) ..controls (304.22bp,747.52bp) and (304.22bp,704.0bp)  .. (304.22bp,704.0bp) .. controls (304.22bp,704.0bp) and (612.79bp,704.0bp)  .. (612.79bp,704.0bp) .. controls (612.79bp,704.0bp) and (612.79bp,702.38bp)  .. (CO);
  \draw [->] (AZ) ..controls (725.0bp,605.75bp) and (725.0bp,351.0bp)  .. (725.0bp,351.0bp) .. controls (725.0bp,351.0bp) and (706.61bp,351.0bp)  .. (NH);
  \draw [->] (KY) ..controls (559.43bp,755.36bp) and (559.43bp,738.0bp)  .. (559.43bp,738.0bp) .. controls (559.43bp,738.0bp) and (705.36bp,738.0bp)  .. (705.36bp,738.0bp) .. controls (705.36bp,738.0bp) and (705.36bp,713.22bp)  .. (AZ);
  \draw [->] (KS) ..controls (242.91bp,1058.2bp) and (242.91bp,997.0bp)  .. (242.91bp,997.0bp) .. controls (242.91bp,997.0bp) and (711.03bp,997.0bp)  .. (OK);
  \draw [->] (VA) ..controls (273.0bp,282.43bp) and (273.0bp,135.0bp)  .. (273.0bp,135.0bp) .. controls (273.0bp,135.0bp) and (290.81bp,135.0bp)  .. (DC);
  \draw [->] (MA) ..controls (328.5bp,226.12bp) and (328.5bp,172.11bp)  .. (DC);
  \draw [->] (OH) ..controls (698.27bp,747.96bp) and (698.27bp,722.0bp)  .. (698.27bp,722.0bp) .. controls (698.27bp,722.0bp) and (471.77bp,722.0bp)  .. (471.77bp,722.0bp) .. controls (471.77bp,722.0bp) and (471.77bp,705.9bp)  .. (LA);
  \draw [->] (ND) ..controls (264.0bp,1183.1bp) and (264.0bp,1143.0bp)  .. (264.0bp,1143.0bp) .. controls (264.0bp,1143.0bp) and (160.0bp,1143.0bp)  .. (160.0bp,1143.0bp) .. controls (160.0bp,1143.0bp) and (160.0bp,1141.2bp)  .. (IA);
  \draw [->] (OH) ..controls (691.18bp,741.98bp) and (691.18bp,706.0bp)  .. (691.18bp,706.0bp) .. controls (691.18bp,706.0bp) and (400.38bp,706.0bp)  .. (400.38bp,706.0bp) .. controls (400.38bp,706.0bp) and (400.38bp,703.92bp)  .. (TN);
  \draw [->] (NC) ..controls (381.88bp,747.14bp) and (381.88bp,715.0bp)  .. (381.88bp,715.0bp) .. controls (381.88bp,715.0bp) and (694.73bp,715.0bp)  .. (694.73bp,715.0bp) .. controls (694.73bp,715.0bp) and (694.73bp,712.2bp)  .. (AZ);
  \draw [->] (KS) ..controls (233.45bp,1066.2bp) and (233.45bp,1015.0bp)  .. (233.45bp,1015.0bp) .. controls (233.45bp,1015.0bp) and (148.57bp,1015.0bp)  .. (VT);
  \draw [->] (KY) ..controls (524.5bp,758.48bp) and (524.5bp,744.0bp)  .. (524.5bp,744.0bp) .. controls (524.5bp,744.0bp) and (391.12bp,744.0bp)  .. (391.12bp,744.0bp) .. controls (391.12bp,744.0bp) and (391.12bp,711.0bp)  .. (TN);
  \draw [->] (TN) ..controls (390.5bp,644.29bp) and (390.5bp,625.0bp)  .. (390.5bp,625.0bp) .. controls (390.5bp,625.0bp) and (467.0bp,625.0bp)  .. (467.0bp,625.0bp) .. controls (467.0bp,625.0bp) and (467.0bp,495.22bp)  .. (PA);
  \draw [->] (ID) ..controls (439.65bp,1101.0bp) and (430.9bp,1101.0bp)  .. (430.9bp,1101.0bp) .. controls (430.9bp,1101.0bp) and (430.9bp,901.0bp)  .. (430.9bp,901.0bp) .. controls (430.9bp,901.0bp) and (429.9bp,901.0bp)  .. (NM);
  \draw [->] (MN) ..controls (456.09bp,766.6bp) and (456.09bp,712.64bp)  .. (LA);
  \draw [->] (SC) ..controls (241.56bp,715.66bp) and (241.56bp,573.0bp)  .. (241.56bp,573.0bp) .. controls (241.56bp,573.0bp) and (264.51bp,573.0bp)  .. (HI);
  \draw [->] (IN) ..controls (344.43bp,748.18bp) and (344.43bp,695.0bp)  .. (344.43bp,695.0bp) .. controls (344.43bp,695.0bp) and (432.8bp,695.0bp)  .. (432.8bp,695.0bp) .. controls (432.8bp,695.0bp) and (432.8bp,694.41bp)  .. (LA);
  \draw [->] (MO) ..controls (616.68bp,747.63bp) and (616.68bp,724.0bp)  .. (616.68bp,724.0bp) .. controls (616.68bp,724.0bp) and (395.75bp,724.0bp)  .. (395.75bp,724.0bp) .. controls (395.75bp,724.0bp) and (395.75bp,708.65bp)  .. (TN);
  \draw [->] (MS) ..controls (44.0bp,643.96bp) and (44.0bp,629.0bp)  .. (44.0bp,629.0bp) .. controls (44.0bp,629.0bp) and (805.3bp,629.0bp)  .. (805.3bp,629.0bp) .. controls (805.3bp,629.0bp) and (805.3bp,594.85bp)  .. (WV);
  \draw [->] (VT) ..controls (168.58bp,987.0bp) and (228.73bp,987.0bp)  .. (228.73bp,987.0bp) .. controls (228.73bp,987.0bp) and (228.73bp,813.4bp)  .. (SC);
  \draw [->] (MT) ..controls (393.97bp,1061.7bp) and (393.97bp,1011.0bp)  .. (393.97bp,1011.0bp) .. controls (393.97bp,1011.0bp) and (156.44bp,1011.0bp)  .. (VT);
  \draw [->] (NJ) ..controls (377.23bp,129.0bp) and (371.6bp,129.0bp)  .. (371.6bp,129.0bp) .. controls (371.6bp,129.0bp) and (371.6bp,33.0bp)  .. (371.6bp,33.0bp) .. controls (371.6bp,33.0bp) and (370.79bp,33.0bp)  .. (NY);
  \draw [->] (OR) ..controls (233.72bp,885.0bp) and (287.5bp,885.0bp)  .. (287.5bp,885.0bp) .. controls (287.5bp,885.0bp) and (287.5bp,783.0bp)  .. (287.5bp,783.0bp) .. controls (287.5bp,783.0bp) and (288.32bp,783.0bp)  .. (IN);
  \draw [->] (WI) ..controls (598.59bp,898.0bp) and (982.0bp,898.0bp)  .. (982.0bp,898.0bp) .. controls (982.0bp,898.0bp) and (982.0bp,821.22bp)  .. (ME);
  \draw [->] (NM) ..controls (417.53bp,862.23bp) and (417.53bp,837.0bp)  .. (417.53bp,837.0bp) .. controls (417.53bp,837.0bp) and (973.0bp,837.0bp)  .. (973.0bp,837.0bp) .. controls (973.0bp,837.0bp) and (973.0bp,820.09bp)  .. (ME);
  \draw [->] (ND) ..controls (306.0bp,1193.1bp) and (306.0bp,1179.0bp)  .. (306.0bp,1179.0bp) .. controls (306.0bp,1179.0bp) and (475.0bp,1179.0bp)  .. (475.0bp,1179.0bp) .. controls (475.0bp,1179.0bp) and (475.0bp,1145.2bp)  .. (ID);
  \draw [->] (WY) ..controls (312.5bp,1061.7bp) and (312.5bp,999.0bp)  .. (312.5bp,999.0bp) .. controls (312.5bp,999.0bp) and (710.8bp,999.0bp)  .. (OK);
  \draw [->] (UT) ..controls (284.83bp,879.0bp) and (261.82bp,879.0bp)  .. (261.82bp,879.0bp) .. controls (261.82bp,879.0bp) and (261.82bp,819.72bp)  .. (SC);
  \draw [->] (MN) ..controls (461.32bp,746.81bp) and (461.32bp,720.0bp)  .. (461.32bp,720.0bp) .. controls (461.32bp,720.0bp) and (753.0bp,720.0bp)  .. (753.0bp,720.0bp) .. controls (753.0bp,720.0bp) and (753.0bp,704.54bp)  .. (AL);
  \draw [->] (OR) ..controls (170.8bp,853.42bp) and (170.8bp,825.0bp)  .. (170.8bp,825.0bp) .. controls (170.8bp,825.0bp) and (777.0bp,825.0bp)  .. (777.0bp,825.0bp) .. controls (777.0bp,825.0bp) and (777.0bp,713.09bp)  .. (AL);
  \draw [->] (NC) ..controls (372.62bp,764.14bp) and (372.62bp,753.0bp)  .. (372.62bp,753.0bp) .. controls (372.62bp,753.0bp) and (312.44bp,753.0bp)  .. (312.44bp,753.0bp) .. controls (312.44bp,753.0bp) and (312.44bp,712.7bp)  .. (DE);
  \draw [->] (WY) ..controls (318.75bp,1083.6bp) and (318.75bp,1078.0bp)  .. (318.75bp,1078.0bp) .. controls (318.75bp,1078.0bp) and (475.78bp,1078.0bp)  .. (475.78bp,1078.0bp) .. controls (475.78bp,1078.0bp) and (475.78bp,928.93bp)  .. (WI);
  \draw [->] (OR) ..controls (153.0bp,855.23bp) and (153.0bp,793.0bp)  .. (153.0bp,793.0bp) .. controls (153.0bp,793.0bp) and (69.526bp,793.0bp)  .. (NV);
  \draw [->] (DE) ..controls (334.8bp,666.0bp) and (338.86bp,666.0bp)  .. (338.86bp,666.0bp) .. controls (338.86bp,666.0bp) and (338.86bp,389.3bp)  .. (CA);
  \draw [->] (KY) ..controls (528.75bp,738.97bp) and (528.75bp,681.0bp)  .. (528.75bp,681.0bp) .. controls (528.75bp,681.0bp) and (495.79bp,681.0bp)  .. (LA);
  \draw [->] (UT) ..controls (339.2bp,857.78bp) and (339.2bp,828.0bp)  .. (339.2bp,828.0bp) .. controls (339.2bp,828.0bp) and (832.2bp,828.0bp)  .. (832.2bp,828.0bp) .. controls (832.2bp,828.0bp) and (832.2bp,814.82bp)  .. (RI);
  \draw [->] (FL) ..controls (342.77bp,454.0bp) and (287.0bp,454.0bp)  .. (287.0bp,454.0bp) .. controls (287.0bp,454.0bp) and (287.0bp,380.51bp)  .. (VA);
  \draw [->] (AR) ..controls (812.2bp,963.55bp) and (812.2bp,923.0bp)  .. (812.2bp,923.0bp) .. controls (812.2bp,923.0bp) and (640.38bp,923.0bp)  .. (640.38bp,923.0bp) .. controls (640.38bp,923.0bp) and (640.38bp,812.19bp)  .. (MO);
  \draw [->] (PA) ..controls (468.5bp,426.36bp) and (468.5bp,405.0bp)  .. (468.5bp,405.0bp) .. controls (468.5bp,405.0bp) and (289.0bp,405.0bp)  .. (289.0bp,405.0bp) .. controls (289.0bp,405.0bp) and (289.0bp,378.06bp)  .. (VA);
  \draw [->] (MN) ..controls (492.45bp,744.26bp) and (492.45bp,669.0bp)  .. (492.45bp,669.0bp) .. controls (492.45bp,669.0bp) and (514.46bp,669.0bp)  .. (MI);
  \draw [->] (OK) ..controls (757.0bp,962.44bp) and (757.0bp,937.0bp)  .. (757.0bp,937.0bp) .. controls (757.0bp,937.0bp) and (991.0bp,937.0bp)  .. (991.0bp,937.0bp) .. controls (991.0bp,937.0bp) and (991.0bp,820.07bp)  .. (ME);
  \draw [->] (OK) ..controls (781.0bp,949.1bp) and (781.0bp,795.0bp)  .. (781.0bp,795.0bp) .. controls (781.0bp,795.0bp) and (812.56bp,795.0bp)  .. (RI);
  \draw [->] (NC) ..controls (419.57bp,750.47bp) and (419.57bp,675.0bp)  .. (419.57bp,675.0bp) .. controls (419.57bp,675.0bp) and (419.9bp,675.0bp)  .. (LA);
  \draw [->] (KS) ..controls (247.64bp,1065.2bp) and (247.64bp,1027.0bp)  .. (247.64bp,1027.0bp) .. controls (247.64bp,1027.0bp) and (381.32bp,1027.0bp)  .. (381.32bp,1027.0bp) .. controls (381.32bp,1027.0bp) and (381.32bp,926.07bp)  .. (NM);
  \draw [->] (AK) ..controls (670.8bp,1076.3bp) and (670.8bp,1005.0bp)  .. (670.8bp,1005.0bp) .. controls (670.8bp,1005.0bp) and (712.31bp,1005.0bp)  .. (OK);
  \draw [->] (SC) ..controls (259.78bp,763.48bp) and (259.78bp,762.0bp)  .. (259.78bp,762.0bp) .. controls (259.78bp,762.0bp) and (821.4bp,762.0bp)  .. (821.4bp,762.0bp) .. controls (821.4bp,762.0bp) and (821.4bp,604.3bp)  .. (WV);
  \draw [->] (ID) ..controls (453.56bp,1084.4bp) and (453.56bp,1068.0bp)  .. (453.56bp,1068.0bp) .. controls (453.56bp,1068.0bp) and (337.5bp,1068.0bp)  .. (337.5bp,1068.0bp) .. controls (337.5bp,1068.0bp) and (337.5bp,926.4bp)  .. (UT);
  \draw [->] (WI) ..controls (461.87bp,861.99bp) and (461.87bp,848.0bp)  .. (461.87bp,848.0bp) .. controls (461.87bp,848.0bp) and (271.27bp,848.0bp)  .. (271.27bp,848.0bp) .. controls (271.27bp,848.0bp) and (271.27bp,814.92bp)  .. (SC);
  \draw [->] (NH) ..controls (668.1bp,282.43bp) and (668.1bp,135.0bp)  .. (668.1bp,135.0bp) .. controls (668.1bp,135.0bp) and (457.13bp,135.0bp)  .. (NJ);
  \draw [->] (ME) ..controls (991.0bp,718.3bp) and (991.0bp,589.0bp)  .. (991.0bp,589.0bp) .. controls (991.0bp,589.0bp) and (777.8bp,589.0bp)  .. (777.8bp,589.0bp) .. controls (777.8bp,589.0bp) and (777.8bp,588.03bp)  .. (TX);
  \draw [->] (WA) ..controls (265.96bp,463.0bp) and (367.2bp,463.0bp)  .. (367.2bp,463.0bp) .. controls (367.2bp,463.0bp) and (367.2bp,345.0bp)  .. (367.2bp,345.0bp) .. controls (367.2bp,345.0bp) and (368.72bp,345.0bp)  .. (GA);
  \draw [->] (TN) ..controls (353.0bp,644.94bp) and (353.0bp,573.0bp)  .. (353.0bp,573.0bp) .. controls (353.0bp,573.0bp) and (351.25bp,573.0bp)  .. (HI);
  \draw [->] (CO) ..controls (649.9bp,661.49bp) and (649.9bp,387.98bp)  .. (NH);
  \draw [->] (WI) ..controls (467.82bp,860.17bp) and (467.82bp,845.0bp)  .. (467.82bp,845.0bp) .. controls (467.82bp,845.0bp) and (312.2bp,845.0bp)  .. (312.2bp,845.0bp) .. controls (312.2bp,845.0bp) and (312.2bp,819.61bp)  .. (IN);
  \draw [->] (AR) ..controls (814.5bp,944.62bp) and (814.5bp,839.0bp)  .. (814.5bp,839.0bp) .. controls (814.5bp,839.0bp) and (395.2bp,839.0bp)  .. (395.2bp,839.0bp) .. controls (395.2bp,839.0bp) and (395.2bp,821.03bp)  .. (NC);
  \draw [->] (IA) ..controls (200.12bp,1101.0bp) and (214.0bp,1101.0bp)  .. (214.0bp,1101.0bp) .. controls (214.0bp,1101.0bp) and (214.0bp,897.0bp)  .. (214.0bp,897.0bp) .. controls (214.0bp,897.0bp) and (277.44bp,897.0bp)  .. (UT);
  \draw [->] (OK) ..controls (714.53bp,983.0bp) and (636.82bp,983.0bp)  .. (636.82bp,983.0bp) .. controls (636.82bp,983.0bp) and (636.82bp,815.21bp)  .. (MO);
  \draw [->] (ID) ..controls (481.33bp,1085.7bp) and (481.33bp,1083.0bp)  .. (481.33bp,1083.0bp) .. controls (481.33bp,1083.0bp) and (841.57bp,1083.0bp)  .. (841.57bp,1083.0bp) .. controls (841.57bp,1083.0bp) and (841.57bp,1035.6bp)  .. (AR);
  \draw [->] (WY) ..controls (331.25bp,1061.6bp) and (331.25bp,1013.0bp)  .. (331.25bp,1013.0bp) .. controls (331.25bp,1013.0bp) and (153.25bp,1013.0bp)  .. (VT);
  \draw [->] (AL) ..controls (756.2bp,657.67bp) and (756.2bp,653.0bp)  .. (756.2bp,653.0bp) .. controls (756.2bp,653.0bp) and (295.5bp,653.0bp)  .. (295.5bp,653.0bp) .. controls (295.5bp,653.0bp) and (295.5bp,602.38bp)  .. (HI);
  \draw [->] (WI) ..controls (551.59bp,888.0bp) and (668.6bp,888.0bp)  .. (668.6bp,888.0bp) .. controls (668.6bp,888.0bp) and (668.6bp,815.04bp)  .. (OH);
  \draw [->] (OH) ..controls (673.25bp,738.31bp) and (673.25bp,675.0bp)  .. (673.25bp,675.0bp) .. controls (673.25bp,675.0bp) and (671.85bp,675.0bp)  .. (CO);
  \draw [->] (NM) ..controls (399.67bp,865.57bp) and (399.67bp,859.0bp)  .. (399.67bp,859.0bp) .. controls (399.67bp,859.0bp) and (534.0bp,859.0bp)  .. (534.0bp,859.0bp) .. controls (534.0bp,859.0bp) and (534.0bp,820.25bp)  .. (KY);
  \draw [->] (WY) ..controls (359.77bp,1113.0bp) and (365.67bp,1113.0bp)  .. (365.67bp,1113.0bp) .. controls (365.67bp,1113.0bp) and (365.67bp,894.0bp)  .. (365.67bp,894.0bp) .. controls (365.67bp,894.0bp) and (365.96bp,894.0bp)  .. (NM);
  \draw [->] (MS) ..controls (121.73bp,671.0bp) and (230.5bp,671.0bp)  .. (230.5bp,671.0bp) .. controls (230.5bp,671.0bp) and (230.5bp,472.0bp)  .. (230.5bp,472.0bp) .. controls (230.5bp,472.0bp) and (364.21bp,472.0bp)  .. (FL);
  \draw [->] (PA) ..controls (481.0bp,421.42bp) and (481.0bp,393.0bp)  .. (481.0bp,393.0bp) .. controls (481.0bp,393.0bp) and (361.0bp,393.0bp)  .. (361.0bp,393.0bp) .. controls (361.0bp,393.0bp) and (361.0bp,377.94bp)  .. (CA);
  \draw [->] (MT) ..controls (433.79bp,1113.0bp) and (438.5bp,1113.0bp)  .. (438.5bp,1113.0bp) .. controls (438.5bp,1113.0bp) and (438.5bp,894.0bp)  .. (438.5bp,894.0bp) .. controls (438.5bp,894.0bp) and (439.07bp,894.0bp)  .. (WI);
  \draw [->] (MS) ..controls (57.5bp,627.82bp) and (57.5bp,561.0bp)  .. (57.5bp,561.0bp) .. controls (57.5bp,561.0bp) and (264.4bp,561.0bp)  .. (HI);
  \draw [->] (ME) ..controls (964.0bp,764.4bp) and (964.0bp,760.0bp)  .. (964.0bp,760.0bp) .. controls (964.0bp,760.0bp) and (726.0bp,760.0bp)  .. (726.0bp,760.0bp) .. controls (726.0bp,760.0bp) and (726.0bp,705.77bp)  .. (AZ);
  \draw [->] (IN) ..controls (348.14bp,753.83bp) and (348.14bp,684.0bp)  .. (348.14bp,684.0bp) .. controls (348.14bp,684.0bp) and (348.78bp,684.0bp)  .. (TN);
  \draw [->] (MT) ..controls (366.04bp,1101.0bp) and (356.33bp,1101.0bp)  .. (356.33bp,1101.0bp) .. controls (356.33bp,1101.0bp) and (356.33bp,901.0bp)  .. (356.33bp,901.0bp) .. controls (356.33bp,901.0bp) and (355.28bp,901.0bp)  .. (UT);
  \draw [->] (RI) ..controls (833.82bp,736.61bp) and (833.82bp,649.0bp)  .. (833.82bp,649.0bp) .. controls (833.82bp,649.0bp) and (308.0bp,649.0bp)  .. (308.0bp,649.0bp) .. controls (308.0bp,649.0bp) and (308.0bp,605.2bp)  .. (HI);
  \draw [->] (TX) ..controls (758.0bp,525.87bp) and (758.0bp,489.0bp)  .. (758.0bp,489.0bp) .. controls (758.0bp,489.0bp) and (409.86bp,489.0bp)  .. (409.86bp,489.0bp) .. controls (409.86bp,489.0bp) and (409.86bp,487.75bp)  .. (FL);
  \draw [->] (MO) ..controls (599.35bp,738.96bp) and (599.35bp,675.0bp)  .. (599.35bp,675.0bp) .. controls (599.35bp,675.0bp) and (598.15bp,675.0bp)  .. (MI);
  \draw [->] (TN) ..controls (397.75bp,662.45bp) and (397.75bp,496.85bp)  .. (FL);
  \draw [->] (AL) ..controls (770.6bp,657.17bp) and (770.6bp,603.1bp)  .. (TX);
  \draw [->] (UT) ..controls (317.6bp,864.48bp) and (317.6bp,856.0bp)  .. (317.6bp,856.0bp) .. controls (317.6bp,856.0bp) and (473.78bp,856.0bp)  .. (473.78bp,856.0bp) .. controls (473.78bp,856.0bp) and (473.78bp,821.04bp)  .. (MN);
  \draw [->] (UT) ..controls (262.74bp,891.0bp) and (211.0bp,891.0bp)  .. (211.0bp,891.0bp) .. controls (211.0bp,891.0bp) and (211.0bp,786.0bp)  .. (211.0bp,786.0bp) .. controls (211.0bp,786.0bp) and (73.755bp,786.0bp)  .. (NV);
  \draw [->] (OK) ..controls (736.14bp,948.85bp) and (736.14bp,812.0bp)  .. (736.14bp,812.0bp) .. controls (736.14bp,812.0bp) and (485.69bp,812.0bp)  .. (485.69bp,812.0bp) .. controls (485.69bp,812.0bp) and (485.69bp,810.58bp)  .. (MN);
  \draw [->] (OR) ..controls (194.2bp,848.73bp) and (194.2bp,779.0bp)  .. (194.2bp,779.0bp) .. controls (194.2bp,779.0bp) and (204.6bp,779.0bp)  .. (SC);
  \draw [->] (NM) ..controls (377.34bp,870.96bp) and (377.34bp,864.0bp)  .. (377.34bp,864.0bp) .. controls (377.34bp,864.0bp) and (64.2bp,864.0bp)  .. (64.2bp,864.0bp) .. controls (64.2bp,864.0bp) and (64.2bp,707.11bp)  .. (MS);
  \draw [->] (LA) ..controls (424.75bp,640.78bp) and (424.75bp,537.0bp)  .. (424.75bp,537.0bp) .. controls (424.75bp,537.0bp) and (205.0bp,537.0bp)  .. (205.0bp,537.0bp) .. controls (205.0bp,537.0bp) and (205.0bp,494.04bp)  .. (WA);
  \draw [->] (GA) ..controls (418.5bp,334.12bp) and (418.5bp,279.76bp)  .. (IL);
  \draw [->] (UT) ..controls (333.8bp,851.65bp) and (333.8bp,814.0bp)  .. (333.8bp,814.0bp) .. controls (333.8bp,814.0bp) and (664.2bp,814.0bp)  .. (664.2bp,814.0bp) .. controls (664.2bp,814.0bp) and (664.2bp,810.19bp)  .. (OH);
  \draw [->] (MT) ..controls (406.62bp,1089.0bp) and (406.62bp,928.24bp)  .. (NM);
  \draw [->] (AK) ..controls (626.13bp,1080.1bp) and (626.13bp,1058.0bp)  .. (626.13bp,1058.0bp) .. controls (626.13bp,1058.0bp) and (419.27bp,1058.0bp)  .. (419.27bp,1058.0bp) .. controls (419.27bp,1058.0bp) and (419.27bp,922.48bp)  .. (NM);
  \draw [->] (OH) ..controls (669.5bp,747.56bp) and (669.5bp,702.0bp)  .. (669.5bp,702.0bp) .. controls (669.5bp,702.0bp) and (578.5bp,702.0bp)  .. (578.5bp,702.0bp) .. controls (578.5bp,702.0bp) and (578.5bp,700.62bp)  .. (MI);
  \draw [->] (SC) ..controls (255.22bp,745.42bp) and (255.22bp,717.0bp)  .. (255.22bp,717.0bp) .. controls (255.22bp,717.0bp) and (743.86bp,717.0bp)  .. (743.86bp,717.0bp) .. controls (743.86bp,717.0bp) and (743.86bp,602.49bp)  .. (TX);
  \draw [->] (WY) ..controls (306.25bp,1095.4bp) and (306.25bp,925.15bp)  .. (UT);
  \draw [->] (IA) ..controls (161.33bp,1064.0bp) and (161.33bp,1022.0bp)  .. (161.33bp,1022.0bp) .. controls (161.33bp,1022.0bp) and (400.3bp,1022.0bp)  .. (400.3bp,1022.0bp) .. controls (400.3bp,1022.0bp) and (400.3bp,928.94bp)  .. (NM);
  \draw [->] (FL) ..controls (405.5bp,440.68bp) and (405.5bp,389.05bp)  .. (GA);
  \draw [->] (VT) ..controls (202.36bp,991.0bp) and (361.0bp,991.0bp)  .. (361.0bp,991.0bp) .. controls (361.0bp,991.0bp) and (361.0bp,783.0bp)  .. (361.0bp,783.0bp) .. controls (361.0bp,783.0bp) and (361.69bp,783.0bp)  .. (NC);
  \draw [->] (AL) ..controls (793.85bp,632.26bp) and (793.85bp,567.0bp)  .. (793.85bp,567.0bp) .. controls (793.85bp,567.0bp) and (794.73bp,567.0bp)  .. (WV);
  \draw [->] (MI) ..controls (573.5bp,629.02bp) and (573.5bp,567.0bp)  .. (573.5bp,567.0bp) .. controls (573.5bp,567.0bp) and (357.21bp,567.0bp)  .. (HI);
  \draw [->] (IN) ..controls (308.33bp,747.22bp) and (308.33bp,713.0bp)  .. (308.33bp,713.0bp) .. controls (308.33bp,713.0bp) and (687.64bp,713.0bp)  .. (687.64bp,713.0bp) .. controls (687.64bp,713.0bp) and (687.64bp,709.58bp)  .. (AZ);
  \draw [->] (LA) ..controls (426.5bp,665.99bp) and (426.5bp,488.04bp)  .. (FL);
  \draw [->] (MA) ..controls (362.88bp,234.0bp) and (389.33bp,234.0bp)  .. (389.33bp,234.0bp) .. controls (389.33bp,234.0bp) and (389.33bp,166.75bp)  .. (NJ);
  \draw [->] (HI) ..controls (327.67bp,553.16bp) and (327.67bp,388.24bp)  .. (CA);
  \draw [->] (OK) ..controls (733.57bp,959.0bp) and (733.57bp,842.0bp)  .. (733.57bp,842.0bp) .. controls (733.57bp,842.0bp) and (390.73bp,842.0bp)  .. (390.73bp,842.0bp) .. controls (390.73bp,842.0bp) and (390.73bp,821.02bp)  .. (NC);
  \draw [->] (MS) ..controls (30.5bp,640.99bp) and (30.5bp,613.0bp)  .. (30.5bp,613.0bp) .. controls (30.5bp,613.0bp) and (462.0bp,613.0bp)  .. (462.0bp,613.0bp) .. controls (462.0bp,613.0bp) and (462.0bp,493.21bp)  .. (PA);
  \draw [->] (SD) ..controls (234.0bp,1187.7bp) and (234.0bp,1161.0bp)  .. (234.0bp,1161.0bp) .. controls (234.0bp,1161.0bp) and (461.5bp,1161.0bp)  .. (461.5bp,1161.0bp) .. controls (461.5bp,1161.0bp) and (461.5bp,1142.6bp)  .. (ID);
  \draw [->] (WI) ..controls (449.96bp,868.22bp) and (449.96bp,853.0bp)  .. (449.96bp,853.0bp) .. controls (449.96bp,853.0bp) and (40.5bp,853.0bp)  .. (40.5bp,853.0bp) .. controls (40.5bp,853.0bp) and (40.5bp,818.88bp)  .. (NV);
  \draw [->] (MO) ..controls (593.7bp,759.66bp) and (593.7bp,742.0bp)  .. (593.7bp,742.0bp) .. controls (593.7bp,742.0bp) and (466.54bp,742.0bp)  .. (466.54bp,742.0bp) .. controls (466.54bp,742.0bp) and (466.54bp,709.54bp)  .. (LA);
  \draw [->] (IN) ..controls (300.11bp,773.42bp) and (300.11bp,712.9bp)  .. (DE);
  \draw [->] (IA) ..controls (168.0bp,1067.4bp) and (168.0bp,1032.0bp)  .. (168.0bp,1032.0bp) .. controls (168.0bp,1032.0bp) and (464.67bp,1032.0bp)  .. (464.67bp,1032.0bp) .. controls (464.67bp,1032.0bp) and (464.67bp,928.57bp)  .. (WI);
  \draw [->] (KY) ..controls (520.25bp,762.95bp) and (520.25bp,749.0bp)  .. (520.25bp,749.0bp) .. controls (520.25bp,749.0bp) and (320.67bp,749.0bp)  .. (320.67bp,749.0bp) .. controls (320.67bp,749.0bp) and (320.67bp,710.33bp)  .. (DE);
  \draw [->] (PA) ..controls (493.5bp,413.66bp) and (493.5bp,351.0bp)  .. (493.5bp,351.0bp) .. controls (493.5bp,351.0bp) and (455.05bp,351.0bp)  .. (GA);
  \draw [->] (MS) ..controls (103.67bp,664.0bp) and (178.6bp,664.0bp)  .. (178.6bp,664.0bp) .. controls (178.6bp,664.0bp) and (178.6bp,495.43bp)  .. (WA);
  \draw [->] (ID) ..controls (492.44bp,1064.6bp) and (492.44bp,1003.0bp)  .. (492.44bp,1003.0bp) .. controls (492.44bp,1003.0bp) and (711.57bp,1003.0bp)  .. (OK);
  \draw [->] (NM) ..controls (413.07bp,858.08bp) and (413.07bp,831.0bp)  .. (413.07bp,831.0bp) .. controls (413.07bp,831.0bp) and (838.4bp,831.0bp)  .. (838.4bp,831.0bp) .. controls (838.4bp,831.0bp) and (838.4bp,818.19bp)  .. (RI);
  \draw [->] (UT) ..controls (323.0bp,846.99bp) and (323.0bp,803.0bp)  .. (323.0bp,803.0bp) .. controls (323.0bp,803.0bp) and (525.0bp,803.0bp)  .. (525.0bp,803.0bp) .. controls (525.0bp,803.0bp) and (525.0bp,802.37bp)  .. (KY);
  \draw [->] (WI) ..controls (479.73bp,873.65bp) and (479.73bp,819.95bp)  .. (MN);
  \draw [->] (CO) ..controls (613.9bp,650.4bp) and (613.9bp,637.0bp)  .. (613.9bp,637.0bp) .. controls (613.9bp,637.0bp) and (428.25bp,637.0bp)  .. (428.25bp,637.0bp) .. controls (428.25bp,637.0bp) and (428.25bp,485.76bp)  .. (FL);
  \draw [->] (NM) ..controls (386.27bp,874.12bp) and (386.27bp,820.11bp)  .. (NC);
  \draw [->] (WV) ..controls (830.0bp,529.42bp) and (830.0bp,501.0bp)  .. (830.0bp,501.0bp) .. controls (830.0bp,501.0bp) and (357.0bp,501.0bp)  .. (357.0bp,501.0bp) .. controls (357.0bp,501.0bp) and (357.0bp,382.77bp)  .. (CA);
  \draw [->] (NC) ..controls (412.29bp,722.89bp) and (412.29bp,579.0bp)  .. (412.29bp,579.0bp) .. controls (412.29bp,579.0bp) and (645.87bp,579.0bp)  .. (CT);
  \draw [->] (RI) ..controls (841.65bp,720.97bp) and (841.65bp,597.0bp)  .. (841.65bp,597.0bp) .. controls (841.65bp,597.0bp) and (698.6bp,597.0bp)  .. (698.6bp,597.0bp) .. controls (698.6bp,597.0bp) and (698.6bp,595.61bp)  .. (CT);
  \draw [->] (UT) ..controls (301.4bp,873.77bp) and (301.4bp,867.0bp)  .. (301.4bp,867.0bp) .. controls (301.4bp,867.0bp) and (60.8bp,867.0bp)  .. (60.8bp,867.0bp) .. controls (60.8bp,867.0bp) and (60.8bp,709.12bp)  .. (MS);
  \draw [->] (WV) ..controls (816.35bp,535.43bp) and (816.35bp,513.0bp)  .. (816.35bp,513.0bp) .. controls (816.35bp,513.0bp) and (273.44bp,513.0bp)  .. (273.44bp,513.0bp) .. controls (273.44bp,513.0bp) and (273.44bp,388.03bp)  .. (VA);
  \draw [->] (NJ) ..controls (410.0bp,116.68bp) and (410.0bp,65.05bp)  .. (MD);
  \draw [->] (UT) ..controls (306.8bp,876.43bp) and (306.8bp,817.68bp)  .. (IN);
  \draw [->] (ME) ..controls (1000.0bp,722.2bp) and (1000.0bp,567.0bp)  .. (1000.0bp,567.0bp) .. controls (1000.0bp,567.0bp) and (877.31bp,567.0bp)  .. (WV);
  \draw [->] (MI) ..controls (560.0bp,645.56bp) and (560.0bp,633.0bp)  .. (560.0bp,633.0bp) .. controls (560.0bp,633.0bp) and (819.1bp,633.0bp)  .. (819.1bp,633.0bp) .. controls (819.1bp,633.0bp) and (819.1bp,603.65bp)  .. (WV);
  \draw [->] (NC) ..controls (414.71bp,747.24bp) and (414.71bp,699.0bp)  .. (414.71bp,699.0bp) .. controls (414.71bp,699.0bp) and (543.57bp,699.0bp)  .. (543.57bp,699.0bp) .. controls (543.57bp,699.0bp) and (543.57bp,698.05bp)  .. (MI);
  \draw [->] (MT) ..controls (412.95bp,1060.6bp) and (412.95bp,1001.0bp)  .. (412.95bp,1001.0bp) .. controls (412.95bp,1001.0bp) and (710.92bp,1001.0bp)  .. (OK);
  \draw [->] (OK) ..controls (617.63bp,993.0bp) and (154.67bp,993.0bp)  .. (154.67bp,993.0bp) .. controls (154.67bp,993.0bp) and (154.67bp,923.1bp)  .. (OR);
  \draw [->] (NC) ..controls (417.14bp,756.9bp) and (417.14bp,729.0bp)  .. (417.14bp,729.0bp) .. controls (417.14bp,729.0bp) and (620.58bp,729.0bp)  .. (620.58bp,729.0bp) .. controls (620.58bp,729.0bp) and (620.58bp,711.5bp)  .. (CO);
  \draw [->] (TX) ..controls (744.5bp,539.42bp) and (744.5bp,525.0bp)  .. (744.5bp,525.0bp) .. controls (744.5bp,525.0bp) and (208.0bp,525.0bp)  .. (208.0bp,525.0bp) .. controls (208.0bp,525.0bp) and (208.0bp,492.68bp)  .. (WA);
  \draw [->] (NV) ..controls (5.6667bp,732.9bp) and (5.6667bp,605.0bp)  .. (5.6667bp,605.0bp) .. controls (5.6667bp,605.0bp) and (383.25bp,605.0bp)  .. (383.25bp,605.0bp) .. controls (383.25bp,605.0bp) and (383.25bp,491.45bp)  .. (FL);
  \draw [->] (AR) ..controls (850.8bp,987.39bp) and (850.8bp,821.08bp)  .. (RI);
  \draw [->] (SD) ..controls (228.0bp,1200.0bp) and (228.0bp,1138.6bp)  .. (KS);
  \draw [->] (IL) ..controls (380.42bp,243.0bp) and (354.5bp,243.0bp)  .. (354.5bp,243.0bp) .. controls (354.5bp,243.0bp) and (354.5bp,169.57bp)  .. (DC);
  \draw [->] (OK) ..controls (691.28bp,985.0bp) and (552.0bp,985.0bp)  .. (552.0bp,985.0bp) .. controls (552.0bp,985.0bp) and (552.0bp,820.22bp)  .. (KY);
  \draw [->] (AK) ..controls (633.26bp,1042.9bp) and (633.26bp,903.0bp)  .. (633.26bp,903.0bp) .. controls (633.26bp,903.0bp) and (511.22bp,903.0bp)  .. (WI);
  \draw [->] (FL) ..controls (373.51bp,445.0bp) and (359.0bp,445.0bp)  .. (359.0bp,445.0bp) .. controls (359.0bp,445.0bp) and (359.0bp,380.71bp)  .. (CA);
  \draw [->] (ID) ..controls (486.89bp,1090.8bp) and (486.89bp,925.69bp)  .. (WI);
  \draw [->] (SC) ..controls (246.11bp,722.41bp) and (246.11bp,617.0bp)  .. (246.11bp,617.0bp) .. controls (246.11bp,617.0bp) and (665.75bp,617.0bp)  .. (665.75bp,617.0bp) .. controls (665.75bp,617.0bp) and (665.75bp,599.1bp)  .. (CT);
  \draw [->] (WI) ..controls (455.91bp,864.08bp) and (455.91bp,850.0bp)  .. (455.91bp,850.0bp) .. controls (455.91bp,850.0bp) and (67.6bp,850.0bp)  .. (67.6bp,850.0bp) .. controls (67.6bp,850.0bp) and (67.6bp,703.92bp)  .. (MS);
  \draw [->] (WI) ..controls (438.99bp,887.0bp) and (434.7bp,887.0bp)  .. (434.7bp,887.0bp) .. controls (434.7bp,887.0bp) and (434.7bp,789.0bp)  .. (434.7bp,789.0bp) .. controls (434.7bp,789.0bp) and (433.31bp,789.0bp)  .. (NC);
  \draw [->] (IA) ..controls (174.67bp,1070.8bp) and (174.67bp,1037.0bp)  .. (174.67bp,1037.0bp) .. controls (174.67bp,1037.0bp) and (810.71bp,1037.0bp)  .. (810.71bp,1037.0bp) .. controls (810.71bp,1037.0bp) and (810.71bp,1031.9bp)  .. (AR);
  \draw [->] (CO) ..controls (622.8bp,608.28bp) and (622.8bp,469.0bp)  .. (622.8bp,469.0bp) .. controls (622.8bp,469.0bp) and (521.47bp,469.0bp)  .. (PA);
  \draw [->] (AR) ..controls (807.6bp,979.21bp) and (807.6bp,966.0bp)  .. (807.6bp,966.0bp) .. controls (807.6bp,966.0bp) and (192.67bp,966.0bp)  .. (192.67bp,966.0bp) .. controls (192.67bp,966.0bp) and (192.67bp,924.76bp)  .. (OR);
  \draw [->] (SC) ..controls (250.67bp,743.5bp) and (250.67bp,711.0bp)  .. (250.67bp,711.0bp) .. controls (250.67bp,711.0bp) and (680.55bp,711.0bp)  .. (680.55bp,711.0bp) .. controls (680.55bp,711.0bp) and (680.55bp,704.08bp)  .. (AZ);
  \draw [->] (AK) ..controls (728.9bp,1107.0bp) and (849.29bp,1107.0bp)  .. (849.29bp,1107.0bp) .. controls (849.29bp,1107.0bp) and (849.29bp,1031.8bp)  .. (AR);
  \draw [->] (VT) ..controls (60.04bp,1005.0bp) and (13.5bp,1005.0bp)  .. (13.5bp,1005.0bp) .. controls (13.5bp,1005.0bp) and (13.5bp,818.61bp)  .. (NV);
  \draw [->] (UT) ..controls (328.4bp,848.02bp) and (328.4bp,806.0bp)  .. (328.4bp,806.0bp) .. controls (328.4bp,806.0bp) and (595.79bp,806.0bp)  .. (595.79bp,806.0bp) .. controls (595.79bp,806.0bp) and (595.79bp,804.96bp)  .. (MO);
  \draw [->] (AR) ..controls (905.37bp,999.0bp) and (1000.0bp,999.0bp)  .. (1000.0bp,999.0bp) .. controls (1000.0bp,999.0bp) and (1000.0bp,816.51bp)  .. (ME);
  \draw [->] (DC) ..controls (369.38bp,141.0bp) and (373.8bp,141.0bp)  .. (373.8bp,141.0bp) .. controls (373.8bp,141.0bp) and (373.8bp,21.0bp)  .. (373.8bp,21.0bp) .. controls (373.8bp,21.0bp) and (374.83bp,21.0bp)  .. (MD);
  \draw [->] (MT) ..controls (387.65bp,1075.7bp) and (387.65bp,1053.0bp)  .. (387.65bp,1053.0bp) .. controls (387.65bp,1053.0bp) and (833.86bp,1053.0bp)  .. (833.86bp,1053.0bp) .. controls (833.86bp,1053.0bp) and (833.86bp,1037.0bp)  .. (AR);
  \draw [->] (MO) ..controls (636.16bp,757.65bp) and (636.16bp,740.0bp)  .. (636.16bp,740.0bp) .. controls (636.16bp,740.0bp) and (708.91bp,740.0bp)  .. (708.91bp,740.0bp) .. controls (708.91bp,740.0bp) and (708.91bp,712.76bp)  .. (AZ);
  \draw [->] (MN) ..controls (487.3bp,754.94bp) and (487.3bp,735.0bp)  .. (487.3bp,735.0bp) .. controls (487.3bp,735.0bp) and (701.82bp,735.0bp)  .. (701.82bp,735.0bp) .. controls (701.82bp,735.0bp) and (701.82bp,713.0bp)  .. (AZ);
  \draw [->] (MO) ..controls (608.89bp,748.83bp) and (608.89bp,726.0bp)  .. (608.89bp,726.0bp) .. controls (608.89bp,726.0bp) and (324.78bp,726.0bp)  .. (324.78bp,726.0bp) .. controls (324.78bp,726.0bp) and (324.78bp,708.13bp)  .. (DE);
  \draw [->] (SC) ..controls (279.44bp,774.0bp) and (283.25bp,774.0bp)  .. (283.25bp,774.0bp) .. controls (283.25bp,774.0bp) and (283.25bp,704.85bp)  .. (DE);
  \draw [->] (NM) ..controls (372.87bp,875.58bp) and (372.87bp,870.0bp)  .. (372.87bp,870.0bp) .. controls (372.87bp,870.0bp) and (27.0bp,870.0bp)  .. (27.0bp,870.0bp) .. controls (27.0bp,870.0bp) and (27.0bp,821.03bp)  .. (NV);
  \draw [->] (RI) ..controls (849.47bp,764.96bp) and (849.47bp,599.75bp)  .. (WV);
  \draw [->] (VT) ..controls (186.08bp,989.0bp) and (291.75bp,989.0bp)  .. (291.75bp,989.0bp) .. controls (291.75bp,989.0bp) and (291.75bp,792.0bp)  .. (291.75bp,792.0bp) .. controls (291.75bp,792.0bp) and (292.51bp,792.0bp)  .. (IN);
  \draw [->] (HI) ..controls (285.0bp,559.81bp) and (285.0bp,382.38bp)  .. (VA);
  \draw [->] (MN) ..controls (482.15bp,751.57bp) and (482.15bp,731.0bp)  .. (482.15bp,731.0bp) .. controls (482.15bp,731.0bp) and (624.47bp,731.0bp)  .. (624.47bp,731.0bp) .. controls (624.47bp,731.0bp) and (624.47bp,712.41bp)  .. (CO);
  \draw [->] (VT) ..controls (74.676bp,993.0bp) and (57.4bp,993.0bp)  .. (57.4bp,993.0bp) .. controls (57.4bp,993.0bp) and (57.4bp,710.77bp)  .. (MS);
  \draw [->] (OH) ..controls (712.45bp,773.79bp) and (712.45bp,712.26bp)  .. (AZ);
  \draw [->] (NM) ..controls (408.6bp,853.56bp) and (408.6bp,820.0bp)  .. (408.6bp,820.0bp) .. controls (408.6bp,820.0bp) and (666.4bp,820.0bp)  .. (666.4bp,820.0bp) .. controls (666.4bp,820.0bp) and (666.4bp,813.11bp)  .. (OH);
  \draw [->] (MN) ..controls (445.63bp,763.19bp) and (445.63bp,751.0bp)  .. (445.63bp,751.0bp) .. controls (445.63bp,751.0bp) and (316.56bp,751.0bp)  .. (316.56bp,751.0bp) .. controls (316.56bp,751.0bp) and (316.56bp,711.87bp)  .. (DE);
  \draw [->] (OR) ..controls (158.0bp,828.94bp) and (158.0bp,678.0bp)  .. (158.0bp,678.0bp) .. controls (158.0bp,678.0bp) and (90.728bp,678.0bp)  .. (MS);
  \draw [->] (UT) ..controls (352.08bp,880.0bp) and (358.67bp,880.0bp)  .. (358.67bp,880.0bp) .. controls (358.67bp,880.0bp) and (358.67bp,774.0bp)  .. (358.67bp,774.0bp) .. controls (358.67bp,774.0bp) and (359.94bp,774.0bp)  .. (NC);
  \draw [->] (AK) ..controls (629.69bp,1064.4bp) and (629.69bp,1007.0bp)  .. (629.69bp,1007.0bp) .. controls (629.69bp,1007.0bp) and (160.43bp,1007.0bp)  .. (VT);
  \draw [->] (CA) ..controls (327.5bp,333.65bp) and (327.5bp,280.31bp)  .. (MA);
  \draw [->] (WY) ..controls (325.0bp,1071.3bp) and (325.0bp,1047.0bp)  .. (325.0bp,1047.0bp) .. controls (325.0bp,1047.0bp) and (826.14bp,1047.0bp)  .. (826.14bp,1047.0bp) .. controls (826.14bp,1047.0bp) and (826.14bp,1036.9bp)  .. (AR);
  \draw [->] (KS) ..controls (252.36bp,1080.7bp) and (252.36bp,1073.0bp)  .. (252.36bp,1073.0bp) .. controls (252.36bp,1073.0bp) and (470.22bp,1073.0bp)  .. (470.22bp,1073.0bp) .. controls (470.22bp,1073.0bp) and (470.22bp,929.12bp)  .. (WI);
  \draw [->] (KS) ..controls (257.09bp,1070.3bp) and (257.09bp,1042.0bp)  .. (257.09bp,1042.0bp) .. controls (257.09bp,1042.0bp) and (818.43bp,1042.0bp)  .. (818.43bp,1042.0bp) .. controls (818.43bp,1042.0bp) and (818.43bp,1035.5bp)  .. (AR);
  \draw [->] (NM) ..controls (365.57bp,887.0bp) and (363.33bp,887.0bp)  .. (363.33bp,887.0bp) .. controls (363.33bp,887.0bp) and (363.33bp,792.0bp)  .. (363.33bp,792.0bp) .. controls (363.33bp,792.0bp) and (361.65bp,792.0bp)  .. (IN);
  \draw [->] (MO) ..controls (640.05bp,766.58bp) and (640.05bp,758.0bp)  .. (640.05bp,758.0bp) .. controls (640.05bp,758.0bp) and (765.0bp,758.0bp)  .. (765.0bp,758.0bp) .. controls (765.0bp,758.0bp) and (765.0bp,711.8bp)  .. (AL);
  \draw [->] (NV) ..controls (92.007bp,772.0bp) and (186.4bp,772.0bp)  .. (186.4bp,772.0bp) .. controls (186.4bp,772.0bp) and (186.4bp,497.18bp)  .. (WA);
  \draw [->] (NC) ..controls (377.25bp,769.32bp) and (377.25bp,713.25bp)  .. (TN);
  \draw [->] (VT) ..controls (116.0bp,935.86bp) and (116.0bp,823.0bp)  .. (116.0bp,823.0bp) .. controls (116.0bp,823.0bp) and (773.0bp,823.0bp)  .. (773.0bp,823.0bp) .. controls (773.0bp,823.0bp) and (773.0bp,713.15bp)  .. (AL);
  \draw [->] (IN) ..controls (342.57bp,723.47bp) and (342.57bp,561.0bp)  .. (342.57bp,561.0bp) .. controls (342.57bp,561.0bp) and (640.47bp,561.0bp)  .. (CT);
  \draw [->] (WI) ..controls (511.25bp,878.0bp) and (543.0bp,878.0bp)  .. (543.0bp,878.0bp) .. controls (543.0bp,878.0bp) and (543.0bp,821.32bp)  .. (KY);
  \draw [->] (OH) ..controls (731.62bp,770.0bp) and (769.0bp,770.0bp)  .. (769.0bp,770.0bp) .. controls (769.0bp,770.0bp) and (769.0bp,712.42bp)  .. (AL);
  \draw [->] (NH) ..controls (649.9bp,292.38bp) and (649.9bp,189.0bp)  .. (649.9bp,189.0bp) .. controls (649.9bp,189.0bp) and (358.75bp,189.0bp)  .. (358.75bp,189.0bp) .. controls (358.75bp,189.0bp) and (358.75bp,166.72bp)  .. (DC);
  \draw [->] (OK) ..controls (738.71bp,940.86bp) and (738.71bp,790.0bp)  .. (738.71bp,790.0bp) .. controls (738.71bp,790.0bp) and (734.15bp,790.0bp)  .. (OH);
  \draw [->] (AZ) ..controls (719.0bp,610.17bp) and (719.0bp,455.0bp)  .. (719.0bp,455.0bp) .. controls (719.0bp,455.0bp) and (525.6bp,455.0bp)  .. (PA);
  \draw [->] (CA) ..controls (350.25bp,307.69bp) and (350.25bp,252.0bp)  .. (350.25bp,252.0bp) .. controls (350.25bp,252.0bp) and (385.37bp,252.0bp)  .. (IL);
  \draw [->] (KY) ..controls (548.86bp,765.17bp) and (548.86bp,711.44bp)  .. (MI);
  \draw [->] (WV) ..controls (843.65bp,506.16bp) and (843.65bp,381.0bp)  .. (843.65bp,381.0bp) .. controls (843.65bp,381.0bp) and (417.75bp,381.0bp)  .. (417.75bp,381.0bp) .. controls (417.75bp,381.0bp) and (417.75bp,379.71bp)  .. (GA);
  \draw [->] (NE) ..controls (488.5bp,1199.1bp) and (488.5bp,1142.7bp)  .. (ID);
  \draw [->] (ME) ..controls (982.0bp,717.23bp) and (982.0bp,593.0bp)  .. (982.0bp,593.0bp) .. controls (982.0bp,593.0bp) and (705.8bp,593.0bp)  .. (705.8bp,593.0bp) .. controls (705.8bp,593.0bp) and (705.8bp,591.64bp)  .. (CT);
  \draw [->] (RI) ..controls (797.09bp,775.0bp) and (746.43bp,775.0bp)  .. (746.43bp,775.0bp) .. controls (746.43bp,775.0bp) and (746.43bp,603.32bp)  .. (TX);
  \draw [->] (TN) ..controls (355.0bp,625.07bp) and (355.0bp,468.0bp)  .. (355.0bp,468.0bp) .. controls (355.0bp,468.0bp) and (233.64bp,468.0bp)  .. (WA);
  \draw [->] (UT) ..controls (344.6bp,863.45bp) and (344.6bp,834.0bp)  .. (344.6bp,834.0bp) .. controls (344.6bp,834.0bp) and (964.0bp,834.0bp)  .. (964.0bp,834.0bp) .. controls (964.0bp,834.0bp) and (964.0bp,816.67bp)  .. (ME);
  \draw [->] (AR) ..controls (809.9bp,972.73bp) and (809.9bp,952.0bp)  .. (809.9bp,952.0bp) .. controls (809.9bp,952.0bp) and (561.0bp,952.0bp)  .. (561.0bp,952.0bp) .. controls (561.0bp,952.0bp) and (561.0bp,816.54bp)  .. (KY);
  \draw [->] (NM) ..controls (428.75bp,880.0bp) and (436.6bp,880.0bp)  .. (436.6bp,880.0bp) .. controls (436.6bp,880.0bp) and (436.6bp,777.0bp)  .. (436.6bp,777.0bp) .. controls (436.6bp,777.0bp) and (437.14bp,777.0bp)  .. (MN);
  \draw [->] (WA) ..controls (190.0bp,434.54bp) and (190.0bp,429.0bp)  .. (190.0bp,429.0bp) .. controls (190.0bp,429.0bp) and (640.8bp,429.0bp)  .. (640.8bp,429.0bp) .. controls (640.8bp,429.0bp) and (640.8bp,384.59bp)  .. (NH);
  \draw [->] (WA) ..controls (235.3bp,450.0bp) and (264.33bp,450.0bp)  .. (264.33bp,450.0bp) .. controls (264.33bp,450.0bp) and (264.33bp,389.16bp)  .. (VA);
  \draw [->] (DC) ..controls (338.0bp,116.68bp) and (338.0bp,65.05bp)  .. (NY);
  \draw [->] (AR) ..controls (816.8bp,937.46bp) and (816.8bp,809.0bp)  .. (816.8bp,809.0bp) .. controls (816.8bp,809.0bp) and (491.64bp,809.0bp)  .. (491.64bp,809.0bp) .. controls (491.64bp,809.0bp) and (491.64bp,807.52bp)  .. (MN);
  \draw [->] (WI) ..controls (583.53bp,893.0bp) and (844.6bp,893.0bp)  .. (844.6bp,893.0bp) .. controls (844.6bp,893.0bp) and (844.6bp,820.19bp)  .. (RI);
  \draw [->] (VA) ..controls (255.0bp,318.09bp) and (255.0bp,297.0bp)  .. (255.0bp,297.0bp) .. controls (255.0bp,297.0bp) and (395.67bp,297.0bp)  .. (395.67bp,297.0bp) .. controls (395.67bp,297.0bp) and (395.67bp,170.53bp)  .. (NJ);
  \draw [->] (TX) ..controls (771.5bp,518.47bp) and (771.5bp,448.0bp)  .. (771.5bp,448.0bp) .. controls (771.5bp,448.0bp) and (520.51bp,448.0bp)  .. (PA);
  \draw [->] (KS) ..controls (238.18bp,1041.6bp) and (238.18bp,903.0bp)  .. (238.18bp,903.0bp) .. controls (238.18bp,903.0bp) and (282.47bp,903.0bp)  .. (UT);
  \draw [->] (AR) ..controls (823.7bp,930.38bp) and (823.7bp,785.0bp)  .. (823.7bp,785.0bp) .. controls (823.7bp,785.0bp) and (735.93bp,785.0bp)  .. (OH);
  \draw [->] (ME) ..controls (973.0bp,728.81bp) and (973.0bp,645.0bp)  .. (973.0bp,645.0bp) .. controls (973.0bp,645.0bp) and (320.5bp,645.0bp)  .. (320.5bp,645.0bp) .. controls (320.5bp,645.0bp) and (320.5bp,603.92bp)  .. (HI);
  \draw [->] (NM) ..controls (381.8bp,868.27bp) and (381.8bp,861.0bp)  .. (381.8bp,861.0bp) .. controls (381.8bp,861.0bp) and (266.55bp,861.0bp)  .. (266.55bp,861.0bp) .. controls (266.55bp,861.0bp) and (266.55bp,818.04bp)  .. (SC);
  \draw [->] (MN) ..controls (450.86bp,758.71bp) and (450.86bp,747.0bp)  .. (450.86bp,747.0bp) .. controls (450.86bp,747.0bp) and (386.5bp,747.0bp)  .. (386.5bp,747.0bp) .. controls (386.5bp,747.0bp) and (386.5bp,712.43bp)  .. (TN);
  \draw [->] (CT) ..controls (668.1bp,524.81bp) and (668.1bp,462.0bp)  .. (668.1bp,462.0bp) .. controls (668.1bp,462.0bp) and (525.69bp,462.0bp)  .. (PA);
  \draw [->] (HI) ..controls (347.62bp,555.0bp) and (369.4bp,555.0bp)  .. (369.4bp,555.0bp) .. controls (369.4bp,555.0bp) and (369.4bp,357.0bp)  .. (369.4bp,357.0bp) .. controls (369.4bp,357.0bp) and (370.71bp,357.0bp)  .. (GA);
  \draw [->] (ND) ..controls (270.0bp,1200.0bp) and (270.0bp,1138.6bp)  .. (KS);
  \draw [->] (CT) ..controls (677.2bp,549.69bp) and (677.2bp,384.67bp)  .. (NH);
  \draw [->] (KY) ..controls (564.71bp,764.54bp) and (564.71bp,756.0bp)  .. (564.71bp,756.0bp) .. controls (564.71bp,756.0bp) and (761.0bp,756.0bp)  .. (761.0bp,756.0bp) .. controls (761.0bp,756.0bp) and (761.0bp,710.24bp)  .. (AL);
  \draw [->] (AK) ..controls (622.56bp,1084.8bp) and (622.56bp,1063.0bp)  .. (622.56bp,1063.0bp) .. controls (622.56bp,1063.0bp) and (343.75bp,1063.0bp)  .. (343.75bp,1063.0bp) .. controls (343.75bp,1063.0bp) and (343.75bp,922.55bp)  .. (UT);
  \draw [->] (IA) ..controls (138.5bp,1097.0bp) and (138.5bp,1029.3bp)  .. (VT);
  \draw [->] (WI) ..controls (535.62bp,883.0bp) and (611.26bp,883.0bp)  .. (611.26bp,883.0bp) .. controls (611.26bp,883.0bp) and (611.26bp,820.88bp)  .. (MO);
  \draw [->] (NM) ..controls (404.13bp,851.58bp) and (404.13bp,817.0bp)  .. (404.13bp,817.0bp) .. controls (404.13bp,817.0bp) and (603.53bp,817.0bp)  .. (603.53bp,817.0bp) .. controls (603.53bp,817.0bp) and (603.53bp,815.22bp)  .. (MO);
  \draw [->] (MO) ..controls (632.26bp,767.98bp) and (632.26bp,713.13bp)  .. (CO);
  \draw [->] (ID) ..controls (459.11bp,1064.7bp) and (459.11bp,1009.0bp)  .. (459.11bp,1009.0bp) .. controls (459.11bp,1009.0bp) and (158.51bp,1009.0bp)  .. (VT);
  \draw [->] (WA) ..controls (254.89bp,459.0bp) and (318.33bp,459.0bp)  .. (318.33bp,459.0bp) .. controls (318.33bp,459.0bp) and (318.33bp,384.68bp)  .. (CA);
  \draw [->] (IL) ..controls (419.5bp,226.12bp) and (419.5bp,172.11bp)  .. (NJ);
  \draw [->] (MI) ..controls (546.5bp,643.43bp) and (546.5bp,621.0bp)  .. (546.5bp,621.0bp) .. controls (546.5bp,621.0bp) and (741.29bp,621.0bp)  .. (741.29bp,621.0bp) .. controls (741.29bp,621.0bp) and (741.29bp,601.19bp)  .. (TX);
  \draw [->] (DE) ..controls (337.49bp,675.0bp) and (340.71bp,675.0bp)  .. (340.71bp,675.0bp) .. controls (340.71bp,675.0bp) and (340.71bp,417.0bp)  .. (340.71bp,417.0bp) .. controls (340.71bp,417.0bp) and (393.25bp,417.0bp)  .. (393.25bp,417.0bp) .. controls (393.25bp,417.0bp) and (393.25bp,386.31bp)  .. (GA);
  \draw [->] (IA) ..controls (181.33bp,1066.1bp) and (181.33bp,995.0bp)  .. (181.33bp,995.0bp) .. controls (181.33bp,995.0bp) and (711.34bp,995.0bp)  .. (OK);
  \draw [->] (LA) ..controls (472.0bp,664.29bp) and (472.0bp,496.67bp)  .. (PA);
  \draw [->] (DE) ..controls (275.49bp,685.0bp) and (268.89bp,685.0bp)  .. (268.89bp,685.0bp) .. controls (268.89bp,685.0bp) and (268.89bp,388.96bp)  .. (VA);
  \draw [->] (SD) ..controls (197.33bp,1171.6bp) and (197.33bp,1113.0bp)  .. (197.33bp,1113.0bp) .. controls (197.33bp,1113.0bp) and (196.26bp,1113.0bp)  .. (IA);
  \draw [->] (KY) ..controls (554.14bp,751.81bp) and (554.14bp,733.0bp)  .. (554.14bp,733.0bp) .. controls (554.14bp,733.0bp) and (628.37bp,733.0bp)  .. (628.37bp,733.0bp) .. controls (628.37bp,733.0bp) and (628.37bp,713.03bp)  .. (CO);
  \draw [->] (NV) ..controls (11.333bp,726.09bp) and (11.333bp,609.0bp)  .. (11.333bp,609.0bp) .. controls (11.333bp,609.0bp) and (457.0bp,609.0bp)  .. (457.0bp,609.0bp) .. controls (457.0bp,609.0bp) and (457.0bp,489.79bp)  .. (PA);
  \draw [->] (OH) ..controls (684.09bp,743.05bp) and (684.09bp,708.0bp)  .. (684.09bp,708.0bp) .. controls (684.09bp,708.0bp) and (328.89bp,708.0bp)  .. (328.89bp,708.0bp) .. controls (328.89bp,708.0bp) and (328.89bp,704.72bp)  .. (DE);
  \draw [->] (IN) ..controls (346.29bp,751.61bp) and (346.29bp,697.0bp)  .. (346.29bp,697.0bp) .. controls (346.29bp,697.0bp) and (538.29bp,697.0bp)  .. (538.29bp,697.0bp) .. controls (538.29bp,697.0bp) and (538.29bp,695.87bp)  .. (MI);
  \draw [->] (AL) ..controls (763.4bp,637.18bp) and (763.4bp,601.0bp)  .. (763.4bp,601.0bp) .. controls (763.4bp,601.0bp) and (691.4bp,601.0bp)  .. (691.4bp,601.0bp) .. controls (691.4bp,601.0bp) and (691.4bp,599.38bp)  .. (CT);
  \draw [->] (AZ) ..controls (684.2bp,652.99bp) and (684.2bp,641.0bp)  .. (684.2bp,641.0bp) .. controls (684.2bp,641.0bp) and (407.43bp,641.0bp)  .. (407.43bp,641.0bp) .. controls (407.43bp,641.0bp) and (407.43bp,496.78bp)  .. (FL);
  \draw [->] (RI) ..controls (787.72bp,780.0bp) and (721.0bp,780.0bp)  .. (721.0bp,780.0bp) .. controls (721.0bp,780.0bp) and (721.0bp,709.16bp)  .. (AZ);
\end{tikzpicture}